\documentclass[a4paper, 10pt]{article} 

\usepackage{cancel}

\usepackage[centertags]{amsmath}
\usepackage{amsmath}
\usepackage[dutch,british]{babel}
\usepackage{amsfonts}
\usepackage{amssymb} 
\usepackage{amsthm}
\usepackage{newlfont}
\usepackage{color}
 \usepackage{bbm}
 \usepackage{float}
 \usepackage{import}
  \usepackage{epsfig}
  \usepackage{authblk}
    \usepackage{dsfont}

\usepackage{stmaryrd}
\usepackage{mathrsfs}
\usepackage{pstricks,pst-node}
\usepackage{fancyhdr}
\usepackage{graphicx}
\usepackage{fancybox}
\usepackage{geometry}
 \usepackage{psfrag}
\usepackage{caption}
\usepackage{subcaption}
\usepackage{cite} 
 
 \usepackage{multirow}

\theoremstyle{plain}
  \newtheorem{theorem}{Theorem}[section]

  \newtheorem{lemma}[theorem]{Lemma}
  
\theoremstyle{definition}

\theoremstyle{remark}

\numberwithin{equation}{section}

 \newcounter{smallarabics}
\newenvironment{arabicenumerate}
{\begin{list}{{\normalfont\textrm{\arabic{smallarabics})}}}
  {\usecounter{smallarabics}\setlength{\itemindent}{0cm}
  \setlength{\leftmargin}{5ex}\setlength{\labelwidth}{4ex}
  \setlength{\topsep}{0.75\parsep}\setlength{\partopsep}{0ex}
   \setlength{\itemsep}{0ex}}}
{\end{list}}

\newcounter{smallroman}

\newcommand{\ben}{\begin{arabicenumerate}}
\newcommand{\een}{\end{arabicenumerate}}

\newcommand\otimesal{\mathop{\hbox{\raise 1.6 ex
  \hbox{$\scriptscriptstyle\mathrm{al}$}
\kern -0.92 em \hbox{$\otimes$}}}}
\newcommand\oplusal{\mathop{\hbox{\raise 1.6 ex
  \hbox{$\scriptscriptstyle\mathrm{al}$}
\kern -0.92 em \hbox{$\oplus$}}}}
\newcommand\Gammal{\hbox{\raise 1.7 ex
\hbox{$\scriptscriptstyle\mathrm{al}$}\kern -0.50 em $\Gamma$}}

\newcommand{\caE}{{\mathcal E}}

\newcommand{\caH}{{\mathcal H}}

\newcommand{\bbC}{{\mathbb C}}

\newcommand{\opunit}{\text{1}\kern-0.22em\text{l}}

\newcommand{\frR}{{\mathfrak R}}

\newcommand{\norm}{ \|}

\newcommand{\iu}{{\mathrm i}}
\renewcommand{\d}{{\mathrm d}}

\newcommand{\beq}{ \begin{equation} }
\newcommand{\eeq}{ \end{equation} }
\newcommand{\bet}{ \begin{theorem} }
\newcommand{\eet}{ \end{theorem} }

\newcommand{\adjoint}{\mathrm{ad}}

\newcommand{\distance}{\mathrm{dist}}

\newcommand{\tr }{\mathrm{tr}}

\usepackage{physics}
\newcommand{\ad}[1]{\textrm{ad}_{#1}}

\newcommand\numberthis{\addtocounter{equation}{1}\tag{\theequation}}
\newcommand{\dg}[1]{^{(#1)}}

\begin{document}

\title{Rigorous and simple results on very slow thermalization, or quasi-localization, of the disordered quantum chain.}
\author[1]{Wojciech De Roeck}
\author[2]{Fran\c{c}ois Huveneers}
\author[1]{Branko Meeus}
\author[1]{Oskar A. Pro\'sniak}
\affil[1]{Instituut Theoretische Fysica, KU Leuven, 3001 Leuven, Belgium}
\affil[2]{Department of Mathematics, King’s College London, Strand, London WC2R 2LS, United Kingdom}

\date{\today }

\maketitle

\abstract{This paper originates from lectures delivered at the summer school  "Fundamental Problems in Statistical Physics XV" in Bruneck, Italy, in 2021.  We give a brief and limited introduction into ergodicity-breaking induced by disorder. As the title suggests, we include a simple yet rigorous and original result: For a strongly disordered quantum chain, we exhibit a full set of quasi-local quantities whose dynamics is negligible up to times of order $\exp{ c (\log W)^{2-\epsilon}}$, with $\epsilon < 1$, $c$ a numerical constant,  and $W$ the disorder strength. Such a result, that is often referred to as "quasi-localization", can in principle be obtained in other systems as well, but for a disordered quantum chain, its proof is relatively short and transparent.}

\vspace{0.5cm}

\section{Introduction}\label{sec: Introduction}

\subsection{"Non-thermalizing" or "ergodicity-breaking" phases} \label{sec: intro}

Thermalization is probably one of the most natural phenomena in the physical world. All of us know, at some level of understanding, that objects tend to exchange heat with each other until they reach the same temperature. The reason why this process is natural and, in most cases, inescapable, is because it is driven by  entropy and the second law of thermodynamics. Therefore, it is also largely independent of any microscopic details, like for example the distinction classical versus quantum. 
It is therefore rather surprising that there exist systems where thermalization does not occur, or only after a surprisingly long time. In recent years, one tends to call this phenomenon "ergodicity-breaking". Obviously, lack of ergodicity is not necessarily the same as "lack of thermalization", but we will not elaborate on the difference here. In the present note, we take "lack of thermalization" to mean that some local observables do not evolve towards the equilibrium value (set by the energy density of the system) at long or infinite time.  Such a crude definition  will suffice for us. However, for systems with multiple locally conserved quantities or with a time-dependent Hamiltonian, that definition would of course need to be refined.    

The word "phases" in the title of this subsection  suggests that the lack of thermalization happens in a robust way, i.e.~not resulting from any fine-tuning. We will stress the notion of robustness later on. 

Finally, in the literature, the stress is mostly on lack of thermalization for \emph{infinite times} (i.e.~thermalization never happens), which is usually referred to as "many-body localization" (MBL), see e.g.\
\cite{Basko2006,gornyi2005interacting,vznidarivc2008many,oganesyan2007localization,pal2010many,ros2015integrals,kjall2014many,Huse2014,luitz2015many, Imbrie2016, Schreiber2015} for early works and \cite{ RevModPhys.91.021001,doi:10.1146/annurev-conmatphys-031214-014726} for reviews.

This is certainly the most interesting phenomenon from a philosophical and mathematical point of view, but it might not be the most relevant phenomenon experimentally. In this note, we focus on a slightly weaker phenomenon: thermalization need not be delayed indefinitely, but simply for a time that grows faster than polynomially in some natural parameter. This is sometimes denoted as "quasi-localization" or "asymptotic localization" and it means that thermalization happens by processes beyond perturbation theory, see \cite{PhysRevB.80.115104,basko2011weak,Huveneers_2013,de2015asymptotic,DeRoeck2013} for results in many-body systems at positive densities and \cite{frohlich1986localization,benettin1988nekhoroshev,poschel1990small,johansson2010kam,wang2009long,fishman2009perturbation,fishman2012nonlinear} for stronger results and conjectures at finite energies (i.e.\ near the ground state if one is in a quantum setting). 

The existence of genuine MBL has been debated lately 
\cite{de2017stability,luitz2017small,vsuntajs2020quantum,sels2022bath,morningstar2022avalanches}, but the existence of "quasi-localization" is not in doubt and it is within reach of mathematical treatment, as this article hopes to illustrate.   Actually, the system that we will treat in detail - the disordered quantum chain - is a system where genuine MBL is actually expected by many authors, but we will not discuss this at all.

\subsection{Partial or complete lack of thermalization}\label{sec: partial or complete}

This distinction is best illustrated by examples. Imagine a system in a glassy phase. It does not thermalize to its real equilibrium state, or at least not on reasonable time-scales, but clearly some dissipative phenomena take place in a glass.  In particular, there is definitely transport of heat throughout the system.  This is hence an example of partial lack of thermalization.   Another such example are systems with emergent conserved quantities.  For concreteness, we can think of the fermionic Hubbard model in the regime where the on-site interaction $U$ is much larger than other local energy scales.  In this regime, the state with two fermions at site $i$ and zero fermions at site $i+1$, has a very different energy from states where both sites $i,i+1$ host one fermion. 
This "lack of resonance" will lead to a remarkable stability of sites with two fermions, also called \emph{doublons}. The decay time of the doublon grows exponential in $U$ and 
one can introduce the doublon number as an \emph{emergent conserved quantity}.  This emergent conserved quantity delays proper thermalization: the system will first appear to thermalize to a state determined by total energy and total doublon number, instead of a state that is unconstrained by doublon number. This apparent thermalization to a non-equilibrium state is usually called "prethermalization", see \cite{KUWAHARA201696,PhysRevX.7.011018,Gulden_2020,mori2016rigorous,abanin2017rigorous, else2017prethermal,else2017prethermalphases, dumitrescu2018logarithmically,huveneers2020prethermalization, else2020long}.
Therefore, also in this example, the lack of thermalization does not exclude dissipative processes. For example, the doublons can carry dissipative currents. 

By complete lack  of thermalization, we mean that there are no dissipative phenomena taking place (on reasonable timescales), except for dephasing and entanglement spreading.   This is of course a very physical but rather vague definition. In practise, we replace it by a mathematical definition that is more precise:

Complete lack of thermalization means that there is a \emph{full} set of quasi-local quasi conserved quantities, also known as quasi-LIOMs (\emph{local integrals of motion}): Any operator that commutes with all these quasi-LIOMs is necessarily in the algebra generated by the quasi-LIOMs. In Section \ref{sec: disordered chain}, we will explicitly derive absence of dissipative transport from the existence of a full set of quasi-LIOMs.

\subsection{Robustness}
We cannot stress too often that we are talking here about a \emph{robust} phenomenon. Non-interacting spins or particles are non-thermalizing, but this is obviously not a deep or interesting observation. It becomes interesting if the absence of thermalization persists upon adding a generic small coupling between the particles or spins.  Quasi-localization precisely means that the coupling between spins or particles induces thermalization and transport only non-perturbatively in the coupling strength. 
In our experience, the question of robustness is often overlooked when describing ergodicity-breaking. In subsection \ref{sec: dressing transformation}, we illustrate how easy it is to cook up non-robust models of ergodicity-breaking.

\subsection{Importance of interactions}
There is a huge body of literature on Anderson localization \cite{Anderson1958,kunz1980spectre,Frohlich1983,aizenman2015random} and the insulator/metal transition, a concept that was not mentioned in the above discussion. The reason is that we deal with many-body interacting systems whereas Anderson localization concerns one-particle systems, and in particular the transition between extended and localized one-particle states. 
From our point of view, non-interacting particles are always non-thermalizing, regardless of whether the one-particle eigenstates are localized (as in subsection \ref{sec: uncoupled}) or extended (as in subsection \ref{sec: integrable}).

\section{Simple but non-robust examples} \label{sec: Examples}

In this section, we introduce some notation and concepts, suitable for dealing with quantum spin systems. Then we list three examples of non-thermalization. We point out the similarities and crucial differences with the topic of the present article. 
In the terminology used in subsections \ref{sec: intro} and \ref{sec: partial or complete}, all three examples are examples of genuine complete lack of thermalization, but all three are non-robust.

\subsection{Setup for spin chains}\label{sec: setup spins}
Let us consider a spin chain, i.e.\ $N$ copies of $1/2$-spins arranged on a line. The Hilbert space is $\caH_N = (\bbC^2)^{\otimes_N}$ and we label spins or sites by  $i \in \{1,\ldots,N\}$. 
The most general operator on $\caH_N$ is a linear combination of operators of the form $A_1\otimes A_2 \otimes \ldots \otimes A_N$ with $A_i$ acting on spin/site $i$. In practice, only \emph{local} operators, i.e.\ operators for which all but a finite number of $A_i$ are equal to identity, are relevant.
We will use the notation $X,Y,Z$ for the $2\times 2$ Pauli matrices $\sigma_x,\sigma_y,\sigma_x$   and we write then $X_i,Y_i,Z_i$ for copies of these matrices acting on the site $i$.  It is customary to use the notation $X_i$ also for $\mathds{1}\otimes X_i \otimes \mathds{1}$, i.e.\ for operators acting on the full Hilbert space $\caH_N$, but such that they act as identity on spins other than $i$.  We say then that $X_i$ is supported in $i$, that $X_iY_j$ is supported in $\{i,j\}$, etc.  In general, the phrase ``$A$ is supported on a set $S$" means that the operator $A$ can be written as a product $A_S \otimes \mathds{1}_{S^c}$ with $A_S$ acting on the $2^{|S|}$-dimensional space that consists of $|S|$ copies of $\bbC^2$ and $\mathds{1}_{S^c}$ standing for identities acting on the complement of $S$. 
 As already said, the relevant class of observables is one that has some locality, i.e.\ a support consisting of a small number of sites, in particular not growing with $N$.   
In practice, one however cannot avoid dealing with quasi-local operators instead of purely local ones. 
Concretely, we will say that $A$ is $(C_A,\kappa)$-exponentially quasi-local around some region $S$ if it can be written as 
\begin{equation} \label{eq: def quasilocal}
A= \sum_{r=0}^{\infty}  A_r, \qquad ||A_r|| \leq C_A e^{-\kappa r},
\end{equation}
where $A_r$ is supported on $S_r$, the $r$-fattening of $S$, i.e.\  $ S_r= \{i | \,  \mathrm{dist}(i,S)\leq r \}$ and $\norm{\cdot}$ is the operator norm $||A||=\mathop{\sup}\limits_{\psi \in \caH_N, ||\psi||_2=1} ||A\psi||_2$ where we wrote $||\cdot||_2$ for the Hilbert space norm.
The Hamiltonians $H$ that we consider will be sums $H=\sum_i H_i$ where the local terms $H_i$   are exponentially quasi-local around the respective sites $i$.
Of course, one could consider a weaker notion of quasilocality, allowing for example interactions that decay like a power law, but for simplicity we restrict to the exponentially quasi-local setup.  
Insightful reviews on the above setup can be found for example in \cite{hastings2010locality,Nachtergaele2019}.

\subsection{Uncoupled spins: the trivial case} \label{sec: uncoupled}
Consider the Hamiltonian
\begin{equation}\label{eq: trivial ham z}
 H_0=h\sum_{i=1}^N Z_i   .
\end{equation} 
This Hamiltonian describes free, i.e.\ non-interacting, spins.   
It is clear that this model does not thermalize, as each of the operators $Z_i$ is a conserved quantity, i.e.\
\begin{equation}\label{eq: trivial conservation}
\dv{t} Z_i= {\iu} [H_0,Z_i]=0 .
\end{equation}
This means that, whatever profile of energy we impose on the system at the initial time, this profile will be exactly conserved in time. More precisely, let the initial wavefunction be $\Psi_0$ and the time-evolved wavefunction $\Psi_t=e^{-{\iu}tH_0}\Psi_0$, then the energy profile
$$
E_i(t)=  \langle \Psi_t | h Z_i| \Psi_t \rangle 
$$
is actually independent of time $t$.  In particular, there is no flow of energy from high to  low energy densities.

\subsection{Uncoupled spins in disguise: dressing transformation} \label{sec: dressing transformation}
We now introduce $A$ to be an arbitrary short-ranged Hamiltonian of the form
$ A=\sum_i A_i$
with $A_i$ local terms supported in the vicinity of $i$. Then, let us define the Hamiltonian as the Heisenberg evolution generated by $A$ and up to a finite time $s$, of the trivial Hamiltonian $H_0$ introduced in \eqref{eq: trivial ham z}:  
$$
H_1= e^{{\iu}  s A}  H_0 e^{-{\iu}  s A}= \sum_i h\tau_i, \qquad \tau_i =   e^{{\iu}  s A}  Z_i  e^{-{\iu}  s A}.
$$
First, we should convince ourselves that $H_1$ is actually a bonafide Hamiltonian as defined in subsection \ref{sec: setup spins}. Indeed, the Lieb-Robinson bound (discussed below, see e.g.\ Lemma \ref{lem: lr for a}) ensures that $\tau_i=e^{\iu  s A}  Z_i  e^{-\iu s A}$ is exponentially quasi-local around site $i$.   
We could write each of the $\tau_i$ as a sum of local terms, as in \eqref{eq: def quasilocal}. 
If one was presented the operators $\tau_i$ and $H_1$ in such a form, one would probably not see anything special about them. It would just be a hopelessly complicated Hamiltonian. 
However, the above definition of $\tau_i$ and $H_1$ reveals  of course a very special property: all these operators $\tau_i$ are mutually commuting, because
$$
[\tau_i,\tau_j]=[e^{{\iu}  s A}  Z_i  e^{-{\iu}  s A}, e^{{\iu}  s A}  Z_j  e^{-{\iu}  s A}] =  e^{{\iu}  s A}   [ Z_i,Z_j ] e^{-{\iu}  s A}=0. 
$$
Consequently, since $H_1=\sum_i h\tau_i$, we have
$$
[H_1, \tau_i] =0, 
$$
similar to what we had in the trivial example \eqref{eq: trivial conservation}, with the only difference that the $\tau_i$ are exponentially quasi-local, rather than strictly local. 
Of course, all of this could have been guessed from the start: $H_1$ just describes uncoupled spins in disguise.

\subsection{Integrable systems}   \label{sec: integrable}
Let us now consider the transverse Ising chain
$$
H_2=\sum_{i=1}^N   hZ_i  + \sum_{i=1}^{N-1}  J X_iX_{i+1}.
$$
This Hamiltonian is not of the type considered in subsection \ref{sec: dressing transformation}, but it has another special property. It can be mapped via a unitary transformation $U$ (the Jordan Wigner transformation combined with Fourier transform) to non-interacting fermions
$$
U H_2U^*= \sum_{k \in \mathbb{T}_N} \omega(k)  c^\dagger_k c_k, \qquad  \mathbb{T}_N=\frac{2\pi}{N}\times \{1,\ldots,N\},
$$
for some set of operators $c^{}_k,c^\dagger_k$ that satisfy the fermionic commutation relations 
$$
\{c_k, c^\dagger_{k'} \}= \delta_{k,k'}, \qquad   \{c_k, c_{k'} \}=   \{c^\dagger_k, c^\dagger_{k'} \}=0,
$$
and some function $\omega$, that is interpreted as the dispersion relation of the fermionic modes with momentum $k$.  
We see again that there is a sense of non-thermalization here. The occupation numbers  $n_k= c^\dagger_{k} c_{k}$ are conserved:  $[H_2,n_k]=0$.
 If we start out with only the right-moving fermion modes (i.e.\ those with $\partial_k\omega(k)\geq 0$), then this will stay so forever. Again the model has $N$ meaningful (cf. the discussion in subsection \ref{sec: locality}) and independent conserved quantities, which prohibit thermalization. 
One usually refers to models like $H_2$ as being "integrable" and, whereas this model is also "non-interacting", this need not be the case.  The study of integrable models is rich and interesting, but we will not pursue it here. 

\subsection{(Quasi-)locality of conserved quantities}\label{sec: locality}

In all three of the examples above, we identified $N$ conserved, independent quantities. What is so special about this? After all, in general a Hamiltonian $H$ on $\caH_N$ is a large matrix that can always be diagonalized by some unitary $U$, and this always provides us with an algebra of conserved quantities (containing in particular the spectral projections of $H$).  However,    the crucial point in the three examples above, is that the conserved quantities have some form of locality. In the first example they are strictly local and in the second example they are exponentially quasi-local. The third example is a bit of an outlier since the mapping from spins to fermions is itself non-local. However, the concerned quantities are sums of local operators in the fermion representation\footnote{As stated, this is not strictly the case because the occupation numbers $n_k$ are sums of operators supported on a few sites at arbitrary distance from each other. However, one can easily choose a set of conserved quantities that is indeed a sum of quasi-local operators.}, which is still a meaningful notion of locality\footnote{Indeed, conserved quantities that are extensive sums of (quasi-)local operators are the most traditional conserved quantities, e.g.\ energy or particle number.}.\\
More importantly, in all these three examples, the conserved quantities provide us with physically meaningful constraints on the evolution of the system. 
This should be contrasted with a generic system, where conserved quantities like spectral projections do not have any sense of locality\footnote{Except of course, conserved quantities that are smooth functions of $H$, like $H^2,H^3,\ldots$. Such conserved quantities do not provide useful constraints on the time-evolution.} and therefore do not give any practically observable constraints on the time evolution.

\subsection{Fragility of the three examples}

None of the three examples described above are robust.   If one adds a small arbitrary perturbation to the Hamiltonian, the described phenomenon disappears. 
By a small perturbation, we could for example mean that 
$$
H \rightarrow H+\epsilon P, \qquad  P=\sum_i P_i, \qquad \epsilon \ll 1,
$$
with each $P_i$ exponentially quasi-local around $i$. 
For the examples of subsections \ref{sec: uncoupled} and \ref{sec: dressing transformation}, the perturbed systems retain a partial lack of thermalisation. Namely, there will be a single emergent quasi-conserved quantity corresponding to the original unperturbed Hamiltonians, completely analogous to the doublon number conservation described in subsection \ref{sec: partial or complete}.  We do not explain this further but we refer to \cite{abanin2017rigorous}.   In the example of subsection \ref{sec: integrable}, we have weakly broken integrability. 
There is no lack of thermalization remaining, according to our definition.
Unfortunately, due to the history of the subject,
it is precisely in weakly perturbed integrable systems, i.e.\ perturbation of the example from subsection \ref{sec: integrable}, that the term "prethermalization" has become widely known, see e.g.\  \cite{PhysRevLett.93.142002,eckstein2009thermalization, PhysRevX.9.021027,bertini2015prethermalization}.  In that context, one means by this, among other things, that approach to genuine equilibrium becomes visible on time scales that appear long in comparison to other time-scales in the system (in particular, timescales governing the prethermalization to a non-equilibrium state). However, that timescale on which genuine equilibration becomes visible is still no longer than of order $1/\epsilon^2$, whereas the present article focuses on timescales that are superpolynomial in $1/\epsilon$.  Therefore, the phenomenology of a weakly broken Hamiltonian of subsection \ref{sec: integrable} has very little overlap with the topic of the present article, see \cite{bastianello2021hydrodynamics} for a recent review.

\section{Results}\label{sec: disordered chain}
We come to the disordered chain, the main topic of this note.  Its Hamiltonian is
$$
H=\sum_{i=1}^N  \left( h_iZ_i  +  JV_i \right),
$$
where the $h_i$ are i.i.d.\ (independent, identically distributed) random numbers drawn uniformly from the unit interval $[0,1]$ and $V_i$ are local operators whose  support is contained in the discrete interval\footnote{Of course, strictly speaking, the support is contained in $[i,i+R-1] \cap \{1,\ldots, N\}$ since the whole model is defined on the interval $\{1,\ldots, N\}$. We will not highlight this explicitly in what follows, since it is irrelevant for our analysis.
} $[i,i+R-1]$, for some fixed range $R$.
We write
$$
H=H_0 + J V, \qquad  H_0= \sum_i h_iZ_i, \qquad V=\sum_i V_i,
$$
and we will take $0<J \ll 1$. Without loss of generality we can then assume that $||V_i||\leq 1$.

\subsection{Quasi-conserved quantities}\label{sec: quasi conserved}
The goal of our analysis is to find a full set of $N$ quantities $\tau_i$ that are conserved for superpolynomial time in $1/J$ and such that most of them are exponentially quasi-local around a single site. Our main results is
\begin{theorem}[Quasi-Localization with superpolynomial lifetime] \label{thm: quasi}
Let the range parameter $R$, introduced above, be fixed, and choose two parameters $\beta,\epsilon\in (0,1)$. 
Then, there is a $J_*>0$ depending on those parameters $R,\beta,\epsilon$ but not on $N$, such that, for $J\leq J_*$, we can construct  $N$ quasi-conserved operators (also known as quasi-LIOMs)
 $\tau_{i=1,\ldots N}$ with the properties listed below.
Set
\begin{equation}
n_* = \left\lfloor(\log(1/J))^{1-\epsilon}\right\rfloor.    \label{n_choice}
\end{equation}
\begin{enumerate} 
\item The $\tau_i$ are spin-operators in the sense that they are unitarily equivalent to the operators $Z_j$
\item They are mutually commuting:  $[\tau_{i'},\tau_i]=0$ for all $i,i'$.
\item They have the following locality property: for each $i$, there is an interval $M_i$ containing site $i$ such that  $\tau_i$ is $(|M_i|,\frac{1-\beta}{4R} \log(1/J))$- exponentially quasi-local around the interval $M_i$.  We call the resonant set $\frR$ the set of sites $i$ for which $|M_i|>1$.
\item They are quasi-conserved: 
$$||[H,\tau_i]||\leq |M_i| e^{-\frac{1-\beta}{4}\log(\frac{1}{J})^{2-\epsilon}}.$$
\item  They form a maximal commutative sub-algebra: Any operator commuting with all $\tau_i$, is itself contained in the algebra generated by the $\tau_i$.
\item 
The event  $i\in \frR$ depends only on the random variables $h_j$ with $j \in  [i-2n_*R+1,i+2n_*R-1]$, and its probability is bounded as
$$
\mathrm{Prob}(i \in \frR)  \leq  J^{\beta/2}.
$$
\item If $|M_j| >1$, then all sites $i$ satisfying $\mathrm{dist}(i,M_j^c) \geq  n_* R$ are elements of $\frR$. We obtain the bound
$$
\mathrm{Prob}(|M_j| \geq \ell)  \leq  e^{-\frac{\beta}{10R}(\log(1/J))^{\epsilon} \ell}.
$$ 
\end{enumerate}
\end{theorem}
As already argued in subsection \ref{sec: locality}, the locality property of the $\tau_i$-operators is crucial. In its absence, the above result is trivial.  The weaker quasi-conservation property 
$$ ||[H,Z_i]||\leq C J,  $$
i.e.\ choosing simply $\tau_i=Z_i$,
follows immediately from the form of the Hamiltonian. It simply reflects that the spin flipping term is of order $J$. In some sense, the merit of the theorem is to find dressed spin operators such that terms flipping these dressed spins, are of much lower order. \\  To get the result in the form announced in the abstract, 
we can observe that from the item 4. it follows that
\begin{equation}
    || \tau_i(t) - \tau_i(0) || \leq \int_0^t\d s\, \norm{\dv{s} \tau_i(s)} = \int_0^t\d s\,\norm{[H,\tau_i(s)]} \leq t |M_i| e^{-\frac{1-\beta}{4}\log(\frac{1}{J})^{2-\epsilon}}.
\end{equation}
Trading the parameter $J$ for the disorder strength $W$, by a simple scaling argument, we get a timescale 
of order $e^{c (\log W)^{2-\epsilon}}$, where we used that $|M_i|$ is a random variable that is typically of order 1, see item 7.   It is interesting to compare this result with the paper \cite{basko2011weak}. The latter paper exhibits a timescale of order $e^{c(\log W)^3}$ for a chain of classical anharmonic disordered oscillators.  
Indeed, as already stressed, the result in Theorem \ref{thm: quasi} is not specific to disordered quantum systems.  In some earlier papers, we derived related results for classical systems and disorder-free systems, see \cite{Huveneers_2013,de2015asymptotic,DeRoeck2013}. However, in all these cases, the proofs were considerably more involved compared to the present paper, and the results were slightly weaker.

\subsection{Slowness of transport}\label{sec: slowness of transport}

The above theorem on quasi-conserved quantities is in some sense the purest characterization of ergodicity-breaking phases. Yet, from a practical point of view, the slowness of transport is probably a more intuitive statement, so let us set the scene for that.
Given a fiducial site $i_*$, we can define a naive left-right splitting of the Hamiltonian
$$
H= H_L +H_R
$$
by 
\begin{equation}\label{eq: naive splitting}
  H_L= \sum_{i \leq i_*} (h_i Z_i + JV_i), \qquad  H_R=H-H_L  .
\end{equation}
The energy current from the left region $\{ i\, | \, i \leq i_*\}$ to the right region $\{ i\, | \, i > i_*\}$ can be defined as 
$$
 J_E= \iu [H,H_L]= -\iu [H, H_R] .
$$
 In general, there is no reason why this current operator would be smaller than of order $J$. Indeed, even if the system was perfectly localized, i.e.\ the quasi-LIOMs $\tau_i$ were conserved exactly, then still $J_E$ would detect the oscillatory motion within the LIOMs, because the nearby LIOMs have support that significantly intersects both the left region $\{ i\, | \, i \leq i_*\}$ and the right region $\{ i\, | \, i > i_*\}$, see Figure \ref{fig: oskar lioms}.
 \begin{figure}[H]
    \centering
    \includegraphics[width=0.8\textwidth]{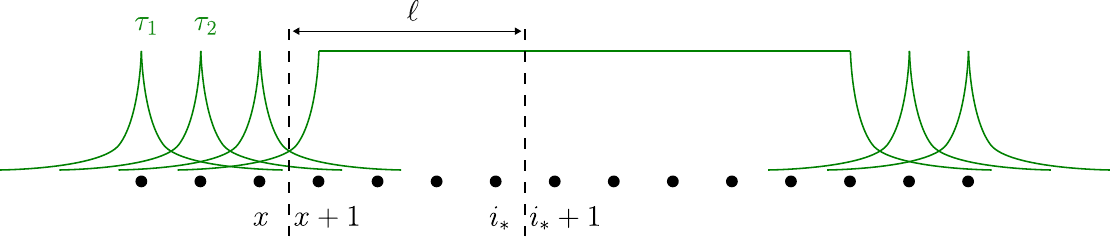}
    \caption{\label{fig: oskar lioms} 
    We represent LIOMs by drawing their peaked local structure. There are LIOMs with large spatial extent at site $i_*$, i.e.\ the interval $M_{i_*}$ is large.  Therefore, to exhibit the smallness of time-integrated currents, we cut the system between sites $x$ and $x+1$, a distance $\ell$ from $i_*$. This results in a larger observable $O$ in Theorem \ref{thm: current}.
    }
\end{figure}

  Such purely oscillatory currents are not persistent, i.e.\ their time integral does not grow in time.  To make this precise, we exhibit that the current $J_E$ is mostly of the form $\iu [H,O]$ for some bounded observable $O$ that is quasi-exponentially local around  $i_*$ or around a short interval containing $i_*$.  If $J_E$ was \emph{exactly} of this form, we would obtain
 \begin{equation}\label{eq: integrated current perfect}
  \int_0^t ds J_E(s) = O(s)-O(0),   
 \end{equation}
 with $J_E(s), O(s), O(0)$ referring to the Heisenberg evolution.  Since the right-hand side is bounded in norm by $2||O||$, then we see that such a time-integrated current is bounded uniformly in time, uniformly in the total size of the system. Note that it was important that $O$ has a norm that does itself not scale with the system size. In particular we could not have done this argument based on the expression  $ J_E= \iu [H,H_L]$ as $H_L$ has a norm proportional to the system size.

The following theorem captures a somewhat less perfect version of this scenario, implying that the time-integrated current $  \int_0^t ds J_E(s)$ is the sum of a term that is bounded uniformly in time, as in \eqref{eq: integrated current perfect}, and a term that can grow in time, but only at a very small rate, namely the bound on the commutator in Theorem \ref{thm: current}.   The presence of this very small rate reflects the fact that we do not show localization for infinite times. 
\begin{theorem}[Slow transport] \label{thm: current} 
Let parameter $ J_*$ be as in Theorem \ref{thm: quasi}. For any $J\leq J_*$,
there is an observable $O$, such that 
$$
 ||J_E-\iu [H,O] || \leq  e^{-\frac{1-\beta}{4} (\log(1/J))^{2-\epsilon} },
$$
such that $||O|| \leq J^{\frac{1-\beta}{8R}} + 2\ell$. 
Here $\ell$ is a random natural number satisfying the bound 
$$
\mathrm{Prob}(\ell\geq \ell_0)  \leq      e^{-c_0(\log(1/J))^{\epsilon}\ell_0}, \qquad c_0= 10^{-3} \beta/R   .$$
\end{theorem}
Let us comment on the interpretation
of random size of the observable $O$. This size reflects the extent of the quasi-LIOMs at or around the site $i_*$, or, in simpler terms, the quality of the localization around site $i_*$. Indeed, if the chain has anomalously weak disorder in a region $I$ with $I \ni i_*$, then, roughly speaking, $O$ is simply the sum of Hamiltonian terms between site $i_*$ and the nearest end of the interval $I$, a distance $\ell$ away from $i_*$, see Figure \ref{fig: oskar lioms}.

In the remainder of this paper, we give the proof of the the above results. First, in Section \ref{sec: cartoon} we present a simplified version of the proof and then, in Sections  \ref{section: Inductive diagonalization},\ref{sec: inductive bounds},\ref{sec: Conclusion} the proper proof.

\section{Cartoon of the proof} \label{sec: cartoon}
The proof is inspired by the seminal work \cite{Imbrie2016}, but since the goal is easier, we need only the very simplest elements of that method, and none of the subtleties that render \cite{Imbrie2016} difficult. From a wider perspective, the method is also known as KAM transformations, often used in classical physics. One could also call it a version of Schrieffer-Wolf transformations.
In the present section, we present a toy version of the proof, with some details missing, and leading to a weakened version of Theorem \ref{thm: quasi}.

\subsection{Perturbative transformation} 

We follow closely the exposition given in \cite{abanin2017effective}.
The rough idea (that will have to be modified later on) is to look for the transformation $U=e^{A}$ with $A=\sum_i A_i$ a sum of exponentially quasi-local anti-Hermitian terms, such that 
$$
H'= e^{\adjoint_A} H = e^{A}H e^{-A} 
$$
is a sum of exponentially quasi-local terms that are moreover diagonal in the basis of $Z_j$. 
That is, in some sense, we anticipate that the disordered model is similar to the example in subsection \ref{sec: dressing transformation}.
We look for such $A$ order-by-order in $J$. In this first step, we anticipate that $A$ is first order in $J$, so expansion in powers of $J$ is the same as expansion in powers of $A$.  We expand, singling out zeroth and first order in $J$:
\begin{equation}\label{def: hprime}
H'= e^{\adjoint_A} H  =  H+ [A,H] + \sum_{n=2} \frac{1}{n!} ( \adjoint_A)^n(H),
\end{equation}
where we wrote $\adjoint_A$ for the superoperator $\adjoint_A=[A,\cdot]$ acting on operators.
The perturbation $JV$ is first order in $J$, and we wish to reduce it to second order in $J$. Therefore, the sum $\sum_{n=2}^{\infty}\ldots$ is for the moment already small enough and we focus on $ H+[A,H]$ which we rewrite as
$$
H_0+JV+ [A,H_0]+  J[A,V]=   H_0 + \left(  JV+  [A,H_0]\right) +  J[A,V]  .   
$$ 
In the last expression $H_0$ is zeroth order, and $ J[A,V]$ is second order. The expression between brackets is first order and it is this expression that we want to eliminate by a judicious choice of $A$, i.e. we look for a solution of the elimination equation
$$
 JV + [A,H_0]=0.
$$
Moreover, the solution $A$ should 
\begin{enumerate}
\item[$i)$] be a sum of local terms $A_i$, 
\item[$ii)$] such that all these terms have strength of at most order $O(J)$, i.e.\ $||A_i||\leq CJ$.
\end{enumerate}
If the first condition is not satisfied, then the new Hamiltonian $H'=e^{\adjoint_A} H$ has no local structure, and it becomes useless (see the discussion in subsection \ref{sec: locality}). If the second condition is not satisfied, then the book-keeping above is no longer correct. 

Finding $A$ involves finding a local inverse to the superoperator $\adjoint_{H_0}$. This can be done since $H_0$ is a sum of local commuting terms: Let $\sigma \in \{-1,1\}^N$ label the eigenvectors of  $(Z_j)_{j=1,\ldots,N}$, hence of $H_0$. They are of the form
$$
|\sigma\rangle = \otimes_{i=1}^N  |\sigma_i\rangle,
$$
where  $\ket{1},\ket{-1}$ stand for the $\pm 1$-eigenvectors of $Z$, in other words for the spin-up/spin-down states, such that 
$$
Z_i |\sigma\rangle =  \sigma_i |\sigma\rangle.
$$  
Now we can find $A$ by writing
\begin{equation}\label{eq: definition A}
\langle \sigma | A |  \sigma'\rangle=   \frac{J \langle \sigma | V |  \sigma'\rangle}{E(\sigma)-E(\sigma')}    \left( 1-\delta_{\sigma \sigma'}\right),
\end{equation}
where $E(\sigma)=\langle \sigma|H_0|\sigma\rangle$.  The factor $\left( 1-\delta_{\sigma \sigma'}\right)$ ensures that $A$ vanishes on the diagonal.  
This solution manifestly satisfies locality. For the sake of explicitness, we pick now a specific form for $V_i$:
$$V_i= X_i + X_{i} X_{i+1}, $$
and we write $A$ as a sum of local terms:
\begin{equation}
A= \sum_i A_i, \quad \langle \sigma | A_i |  \sigma'\rangle=   \frac{J \langle \sigma | V_i |  \sigma'\rangle}{E(\sigma)-E(\sigma')}    \left( 1-\delta_{\sigma \sigma'}\right). \label{eq: definition A refined}
\end{equation}
Obviously, the right-hand side of \eqref{eq: definition A refined} is then nonzero only for $\sigma,\sigma'$ that coincide on all sites except for the sites $i,i+1$, and this property is then inherited by the left-hand side. This is precisely what we mean by locality. Hence condition $i)$ above is manifestly satisfied.
However, condition $ii)$ can fail if the denominator is too small.  
For ease of writing, let us denote by $\sigma^i$ the configuration $\sigma$ flipped at site $i$, i.e.\ with the value (spin) flipped at site $i$, i.e.\ such that $X_i|\sigma\rangle= |\sigma^i\rangle $. Similarly, let $\sigma^{i,j}$ be $\sigma$ flipped at both sites $i,j$.
The choices of $\sigma,\sigma'$ for which  $\langle \sigma | V_i |  \sigma'\rangle$ does not vanish are then
$$
\sigma' =\sigma^{i} \qquad \text{and}   \qquad    \sigma' =\sigma^{i,i+1}.
$$
These lead to the denominators
$$
E(\sigma)-E(\sigma^{i})=  \pm 2h_i  \qquad \text{and}   \qquad E(\sigma)-E(\sigma^{i,i+1})
= \pm 2h_i \pm 2h_{i+1},
$$
where the choice of $+/-$ depends on $\sigma$. If we place a fixed lower bound, say $\delta$, on 
$|E(\sigma)-E(\sigma')|$, then it is clear that we cannot expect this to be satisfied everywhere, i.e. for every $i$. The best we can ask for is that such a non-resonance condition
\begin{equation}\label{eq: non-resonance condition abstract}
|E(\sigma)-E(\sigma')| \geq \delta
\end{equation}
is satisfied at most places in the chain, i.e.\ for most $i$, with $\sigma'$ as above.
The non-resonance condition holds for site $i$ iff\
\begin{equation} \label{eq: non-resonance condition concrete}
2h_i \geq \delta \qquad \text{and} \qquad   |2h_i - 2h_{i+1}|\geq \delta.
\end{equation}
We say that the interval $[i,i+1]$ is resonant if \eqref{eq: non-resonance condition concrete} is violated. 
The probability that a given interval $[i,i+1]$ is resonant is clearly bounded by $C\delta$.
In what follows, we do not want to treat $\delta$ simply as a small constant $c$ and we prefer to think of it as coupled to $J$ as $\delta=J^{\beta}$ with $0<\beta< 1$, but for clarity and overview we will still keep different symbols.  In any case, since $\delta$ is not a mere constant, we  need to revise the condition    $ii)$ above to read
\begin{enumerate}
\item[$ii)'$] All terms in $A$ have strength at most $CJ/\delta$.
\end{enumerate}

However, we still face the problem that, on average, a fraction of order $C\delta$ of sites will violate the non-resonance condition \eqref{eq: non-resonance condition concrete}. We address this issue now.

 \begin{figure}[H]
    \centering
    \includegraphics[width=\textwidth]{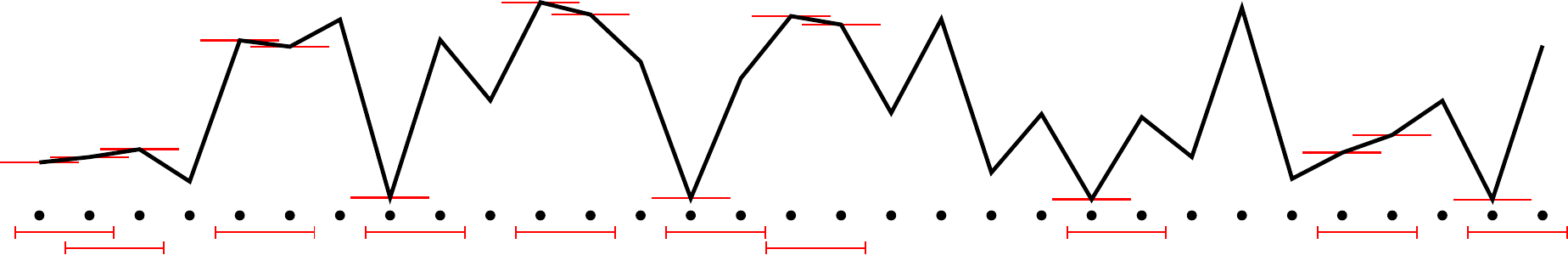}
    \caption{Illustration of resonances. The red horizontal bars denote sites that are resonant because the corresponding fields violate \eqref{eq: non-resonance condition concrete}. By our conventions, we declare the full support of a term $V_i$ resonant, i.e.\ the pair $(i,i+1)$, as indicated at the bottom of the figure.     }
    \label{fig:resonances}
\end{figure}

\subsection{Treatment of resonant regions}\label{sec:treatment of resonant regions}
We have to exclude resonant terms $V_i$, corresponding to resonant intervals $[i,i+1]$, from our procedure and so we split

$$
V=V^{\cancel{\frR}}+V^{\frR}, \qquad   V^{\cancel\frR}=\sum_{i:\, [i,i+1]\, \text{non-resonant}}  V_i, \qquad   V^{\frR}=\sum_{i:\, [i,i+1] \, \text{resonant}}V_i.
$$
 Instead of defining $A$ by the condition \eqref{eq: definition A}, we now define it as $A=\sum_{i\in \cancel{\frR}}  A_i$ with 
$$
\langle \sigma | A_i |  \sigma'\rangle=    \frac{J\langle \sigma | V_i |  \sigma'\rangle}{E(\sigma)-E(\sigma')}\left( 1- \delta_{\sigma \sigma'} \right), \qquad   i \in \cancel{\frR},
$$
i.e.\ the definition applies only to non-resonant sites $i$.   We then have
\begin{equation}\label{def: hprime revisited}
H'= e^{\adjoint_A} H =  H_0+JV^\frR+(JV^{\cancel{\frR}}+ [A,H_0]) +  J  [A,V]+ \sum_{n=2} \frac{1}{n!} (\adjoint_A)^n(H),
\end{equation}
and the choice of $A$ is now such that the term between brackets vanishes. This means that we have obtained 
\begin{equation}\label{eq: omit expression for hprime}
H'= H_0+JV^{\frR}+ J^2  V^{(2)} + \text{higher order terms in $J$},  
\end{equation}
with
$$
J^2 V^{(2)}= J  [A,V] + \frac{1}{2} (\adjoint_A)^2(H)= J  [A,V^{\frR}] +  \frac{J}{2}  [A,V^{\cancel{\frR}}],
$$
where the last equality follows by the definition of $A$. A little thought shows that $J^2 V^{(2)}$ is a sum of local terms supported on  3 consecutive sites, and with  strength $CJ^2/\delta$, i.e.
$$
V^{(2)}=\sum_i V^{(2)}_i, \qquad || V^{(2)}_i ||\leq C/\delta, \qquad V^{(2)}_i \text{ supported on } (i,i+1,i+2).
$$
The higher order terms in \eqref{eq: omit expression for hprime}  have larger support, and we will need to keep track of this. However, in the next subsection \ref{sec: lioms one}, we will for simplicity disregard these terms.  Taking them into account would not invalidate any of the bounds exhibited below (it can be absorbed by the constants $C$ that we write). 

\subsection{Construction of quasi-LIOMs}\label{sec: lioms one}

Let us take stock of what we achieved by the transformation $H\mapsto H'$, i.e.\ what we can say about slowness of transport and thermalization of the Hamiltonians $H'$ and $H$. 

We argue that we have obtained a toy-version of Theorem \ref{thm: quasi}, but instead of superpolynomial lifetime of the quasi-LIOMs, we have a polynomial lifetime, as we will exhibit below. 

Above, we identified resonant intervals $[i,i+1]$. 
We now additionally define the resonant set of sites $\frR$ as the union of all the resonant intervals.    
We partition $\frR$ into  intervals by the following connectedness property: We draw an edge between adjacent elements $i,i+1$ of $\frR$ whenever the interval $[i,i+1]$ is resonant. This defines a graph on $\frR$ and we call the connected components of this graph $T_\alpha$, for an abstract index $\alpha$. Note that necessarily $|T_\alpha|\geq 2$. 
The point of this construction is the following:  For each $\alpha$, we can define operators $ \caE_\alpha$ as 
 $$
 \caE_\alpha =  \sum_{i \in T_\alpha} h_i Z_i  +   J \sum_{i: \{i,i+1\} \subset T_\alpha}
  V^{\frR}_i
 $$
having the following properties:
\begin{enumerate}
    \item    $ H'= \sum_{i: i \not \in \frR}  h_i Z_i + \sum_\alpha   \caE_\alpha +  J^2V^{(2)}, $
 \item $[ \caE_\alpha, \caE_{\alpha'}]=0 $, and $[Z_i,\caE_\alpha]=0$,
  for $i \not\in\frR$,  
 \item $[\caE_\alpha, H'] =J^2[\caE_\alpha, V^{(2)}]  $,    and  $[Z_i, H'] =J^2[Z_i, V^{(2)}]  $,
  for $i \not\in\frR$. 
\end{enumerate}
Indeed, the first two properties follow easily from the above construction. Using these properties, and the fact that $[\caE_\alpha, Z_i]=0$ for $i\not\in\frR$, implies the third property. 
We then also deduce 
 \begin{align}\label{eq: slowness of prelioms 1}
 || [\caE_\alpha, H']  || & \leq  2 J^2||\caE_\alpha ||   \sum_{j: [j,j+2]  \cap T_\alpha\neq \emptyset}   || V^{(2)}_j ||    \leq  C |T_\alpha|\frac{J^2}{\delta}   ||\caE_\alpha ||  ,  \\ 
  || [Z_i, H']  || & \leq  2 J^2||Z_i ||   \sum_{j: [j,j+2] \ni i}   || V^{(2)}_j ||    \leq  C \frac{J^2}{\delta}, \qquad i \not\in\frR     .
 \end{align}
This shows that the operators $\caE_\alpha, Z_i$ for $i \not\in\frR$ are quasi-conserved w.r.t.\ to the transformed Hamiltonian $H'$, at least for a longer time than one would naively infer from the original Hamiltonian.
However, for the moment, the number of these quantities is smaller than $N$, so this is not yet a full set of maximal quasi-conserved quantities. 
We remedy this now.

Fix $\alpha$ and recall that $\caE_\alpha$ is supported on $T_\alpha$. We therefore identify $\caE_\alpha$ with a local operator acting on a $d$-dimensional space, with $d=2^{|T_\alpha|}$. By diagonalizing this operator, one can construct $|T_\alpha|$ mutually commuting spin-operators (in the sense of Theorem \ref{thm: quasi}) $\tau'_{\alpha,j}, j=1,\ldots, |T_\alpha|$ such that $[\tau'_{\alpha,j},\tau'_\alpha ]=0$ and such that $\tau'_{\alpha,j}$ are supported on the interval $T_\alpha$ as well. The labelling by $j$ is arbitrary. 
We now set 
\begin{equation}
 \tau'_i= \begin{cases} Z_{i} &  i \not\in\frR \\
 \tau'_{\alpha,j} &  i \in T_\alpha, \quad|T_\alpha|>1, \quad i=\min T_\alpha +j-1 
 \end{cases}
\end{equation}
obtaining a set of $N$ spin-operators.

For convenience, and to increase the similarity to Theorem \ref{thm: quasi}, we now associate to every $\tau_i'$ an interval $M_i$. If $i\in\frR$, then we set $M_i=T_\alpha$ and for $i\not\in\frR$, we simply take $M_i=\{i\}$.  With these definitions, the spin operator $\tau_i'$ is supported on $M_i$. 

Proceeding now as in the derivation of \eqref{eq: slowness of prelioms 1} and using that now $\tau'_i$  have unit norm, we obtain
 \begin{equation}\label{eq: slowness of prelioms 2}
 || [\tau'_i, H^{'}]  || \leq C |M_i|\frac{J^2}{\delta}      .
 \end{equation}

 Since the majority of sites are non-resonant, most of the intervals $M_i$ above contain a single site.
 In contrast, how likely is it that a given site $i$ is the first left-most site of an interval $T_\alpha$ with large length $L$? For that to happen, the intervals $[j,j+1]$, for $j=i,\ldots, i+L-2$, all need to be resonant. The probability that one of these intervals is resonant, is bounded by $C\delta$.
 Since, for a given interval $[j,j+1]$, the resonance condition depends only on $h_j,h_{j+1}$, the resonance conditions are independent whenever the intervals are disjoint. Therefore, the probability that $i$ is the left-most site of $T_\alpha$ with $|T_\alpha|=L$, is bounded by $(C\delta)^{L/2}$. 
The picture that emerges is that long intervals $T_\alpha$ are exponentially sparse.

Finally, we note that the constructed quantities are almost conserved with respect to $H'$ and not with respect to $H$, but this is of course fixed by undoing the dressing transformation and setting
$$
\tau_i = e^{-\adjoint_{A}} \tau'_i .
$$
By the Lieb-Robinson bound (discussed below, see e.g.\ Lemma \ref{lem: lr for a}), this transformation does not spoil the exponential quasi-locality properties.   Also, since the transformation $e^{-\adjoint_{A}}$ preserves norms and products, we obtain from \eqref{eq: slowness of prelioms 2}
that 
 $$  || [\tau_i, H]  ||  =  || [\tau'_i, H^{'}]  || \leq C |M_i|\frac{J^2}{\delta}  .
 $$  

This completes the first step of our procedure.
We started with the operators $Z_i$, that were conserved up to time of order $1/J$, because $||[H,Z_i]||\leq C J$, and we obtained operators $\tau_i$, that are conserved up to times of order $J^2/\delta$.  As already mentioned, this provides a weakened version of Theorem \ref{thm: quasi}.  To obtain the full Theorem \ref{thm: quasi}, we will have to do the above procedure up to a high order in $J$, which will lead to combinatoric challenges.

\subsection{Roadmap of the proof} Before diving into the technicalities of the proper proof, let us present a roadmap.
In Section \ref{section: Inductive diagonalization}, we construct the unitary $e^{A}$ that transforms the Hamiltonian $H$ into a quasi-diagonal Hamiltonian $H'$.  The definition of $A$ proceeds iteratively in orders of $J$.  This part is fairly explicit and algebraic in nature.
 In Section \ref{sec: inductive bounds}, we derive general inductive bounds on the local terms of $A$, and other derived operators. This is kept relatively easy because the bounds we derive are far from sharp. Still, this is probably the least accessible part of the proof. 
 
 Finally, in Section \ref{sec: Conclusion}, we fix the order $n_*$ in $J$ up to which we proceed with the diagonalization. This order must be high enough so that the non-diagonal part is sufficiently small, but small enough so that we can still easily show that resonances are sparse. Once $n_*$ is fixed, it remains to actually construct the quasi-LIOMs and derive locality bounds on them. The latter part has no genuinly new ideas compared to the cartoon of the proof in Section \ref{sec: cartoon}. Once the quasi-LIOMs are constructed, Theorem \ref{thm: quasi}
 is proven and Theorem \ref{thm: current} follows rather straightforwardly.

\section{Inductive quasi-diagonalization of $H$} \label{section: Inductive diagonalization}
We will now present the proper proof. We recall that our parameters are $J,\delta$ and the range $R$. At some point, we will write $\delta=J^{\beta}$ with $\beta \in (0,1)$. 
We consider $H = H_0 + J V$, with $H_0 = \sum h_i Z_i$. By our conventions, both $H_0$ and $V$ are sums of local terms
\begin{align}
    &H_0=\sum_i (H_0)_i, \qquad  V=\sum_i V_i,  \\
    &||(H_0)_i||, ||V_i|| \leq 1.
\end{align}

We will construct an anti-Hermitian generator $A = \sum_{m =1}^{n} J^m A_m$ such that
\begin{equation} \label{eq: result of scheme}
    H' = e^{\textrm{ad}_A} H = H_0 + \sum_{m=1}^n J^m D\dg{m} + \sum_{m=1}^n J^m W^{(m)\frR} +\delta H^{(n)},
\end{equation}
where:
\begin{enumerate}
\item All operators on the right hand side are  sums of local terms.
\item  $D\dg{m}$ is diagonal in $Z$-basis.
\item $W^{(m)\frR}$ is sparse in space: only a small fraction of the local terms is non-zero.
\item  $\delta H^{(n)}$ has only terms of order at least $n+1$ in $J$. 
\end{enumerate}

 We posit that $A = \sum_{m = 1}^{n} J^m A\dg{m}$ and we expand $e^{\textrm{ad}_A} H$  in orders of $J$;
\begin{equation}
    e^{\textrm{ad}_A} H = H_0 + \sum_{m = 1}^\infty J^m G\dg{m}, \label{eq: series}
\end{equation}
where
\begin{align*}
    G\dg{1} &= \ad{A\dg{1}} H_0 + V, \\
    G\dg{m} &= \sum_{k = 1}^{\infty} \frac{1}{k!} \sum_{\substack{ 1 \leq m_1, \ldots, m_k \leq \min(m,n)\\
                  m_1 + \ldots + m_k = m}}  \ad{A\dg{m_{k}}} \ldots \ad{A\dg{m_1}} H_0 \\
            &+ \sum_{k = 1}^{\infty} \frac{1}{k!} \sum_{\substack{ 1 \leq m_1, \ldots, m_k \leq \min(m - 1,n)\\
                  m_1 + \ldots + m_k = m-1}}  \ad{A\dg{m_k}} \ldots \ad{A\dg{m_1}} V, \quad m>1. \numberthis \label{eq: G}
\end{align*}
We note that the first sum over $k$ can be restricted to run from $1$ to $m$ and the second to run from $1$ to $m-1$. 
The convergence of the series in \eqref{eq: series} follows easily since the chain has finite length $N$. Later we will prove a bound that is uniform in the chain length. 
For $m\leq n$, we now rewrite the series for $G\dg{m}$. We single out the term with $A\dg{m}$ and we denote the remainder of $G\dg{m}$ by $F\dg{m}$:
\begin{equation}
    G\dg{m} = \comm{A\dg{m}}{H_0} + F\dg{m}.
\end{equation}
We note from \eqref{eq: G} that $F\dg{m}$ contains only  $A\dg{m'}$ with $m' < m$, which is crucial for setting up the induction.
All the above operators can be written as sums of local terms. We make this explicit as follows:
To any scale $m \leq n$, we associate a range equal to $mR$. Then $A\dg{m}=\sum_i A\dg{m}_i$
such that every local term $A\dg{m}_i$ is supported in the interval $S_i\dg{m}=[i,i+1,\ldots,i+mR-1]$. The possibility of such a representation will be shown below. 
For the operators $F\dg{m}$, we then fix the local decomposition in terms of the local decomposition by $A$ of lower scale as follows: 
\begin{align}
    F\dg{1}_i &= V_i,  \nonumber\\
    F\dg{m}_i &= \sum_{k = 2}^{m} \frac{1}{k!} \sum_{\substack{ 1 \leq m_1, \ldots, m_k \leq m-1\\
                  m_1 + \ldots + m_k = m \\
                  i=\min(i_0,i_1,\ldots,i_k)
                  }}  \ad{A\dg{m_{k}}_{i_k}} \ldots \ad{A\dg{m_1}_{i_1}} (H_0)_{i_0} \nonumber \\
            &+ \sum_{k = 1}^{\infty} \frac{1}{k!} \sum_{\substack{ 1 \leq m_1, \ldots, m_k \leq m - 1\\
                  m_1 + \ldots + m_k = m - 1\\
                  i=\min(i_0,i_1,\ldots,i_k)}}  \ad{A\dg{m_k}_{i_k}} \ldots \ad{A\dg{m_1}_{i_1}} V_{i_0}, \quad m>1. \label{eq: f expansion}
\end{align}

\begin{lemma} If  $A\dg{m'}_i$ has support in in the interval $S_i\dg{m'}$, for every $m'<m$, then $ F\dg{m}_i$ has support inside the interval 
$S_i\dg{m}$.
\end{lemma}
\begin{proof}
Observe that for any local operators $B,B'$ if $[B',B]\neq 0$, then the supports of $B$ and $B'$ overlap and the support of $[B,B']$ is contained within their union. 
\end{proof}

To proceed, we  split $F\dg{m}_i$ into terms that will be eliminated, and terms that cannot be eliminated, either because they are diagonal in $Z$-basis, or because they are resonant. 
We write $D\dg{m}_i $ for the part of $F\dg{m}_i$ that is diagonal in $Z$-basis. Next, we need to split the terms $F\dg{m}_i-D\dg{m}_i$ into resonant and non-resonant terms. We will do this in a crude way.
We say an interval $S_i\dg{m}$ is non-resonant at scale $m$ whenever 
\begin{equation}
\min_{\eta \neq 0}\left( \abs{\sum_{i \in S_i\dg{m}} \eta_i h_i} \right)\geq \delta, \qquad  \delta = J^\beta \label{eq:non-resonance}
\end{equation}
where the minimum is taken over $ \eta \in \{-1,0,1\}^{mR}$, $\eta \neq 0$, and $\beta \in (0,1)$ is a parameter. 
\begin{figure}[H]
    \centering
    \includegraphics[width=0.6\textwidth]{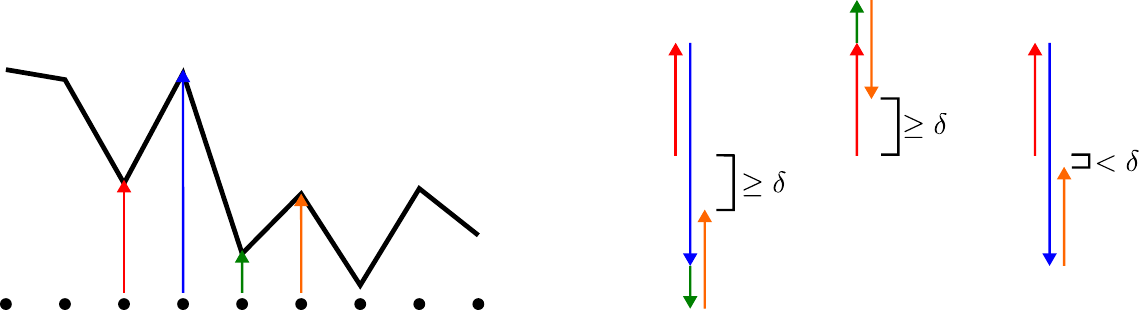}
    \caption{Illustration of resonances for $|S_i\dg{m}|=4$ i.e.\ involving four consecutive sites.  The lengths of arrows correspond to on-site random energies $h_i$ and their orientation to the choice of $\eta_i=\pm 1$. If $\eta_i=0$ then the arrow is omitted.  We see that the displayed interval is resonant by virtue of the third arrangement of arrows on the right.}
    \label{fig:resonances2}             
\end{figure} Note that, if an interval $S_i\dg{m} $ is resonant at scale $m$ (i.e.\  it fails to be non-resonant), then any intervals $S_{i'}\dg{m'} $ with $m'>m$ and such that   $S_{i}\dg{m} \subset S_{i'}\dg{m'} $  are necessarily also resonant. This is a consequence of using the same resonance threshold $\delta$ at all scales.  
We can now define
\begin{align}
  W^{(m)\cancel\frR}_i& = \begin{cases}
  F\dg{m}_i- D\dg{m}_i     &  \text{ $S_i\dg{m}$ is non-resonant at scale $m$} \\
  0 &  \text{otherwise}
  \end{cases} \\[2mm]
   W^{(m)\frR}_i& = \begin{cases}
0  &  \text{ $S_i\dg{m}$ is non-resonant at scale $m$} \\
  F\dg{m}_i- D\dg{m}_i     & \text{otherwise}
  \end{cases}  
\end{align}

We obtain the following decomposition of $G\dg{m} $.
\begin{equation}
    G\dg{m} = \comm{A\dg{m}}{H_0} + D\dg{m} + W^{(m)\cancel\frR} + W^{(m)\frR}.
\end{equation}
It remains to define $A\dg{m}$ recursively which we do through the elimination equation
\begin{equation}
W^{(m)\cancel\frR}_i + \comm{A\dg{m}_i}{H_0} =0.
\end{equation}
Thanks to the nonresonance condition, this equation has the following solution
\begin{equation}
    \bra{\sigma} A\dg{m}_i \ket{\sigma'} = \frac{\bra{\sigma} W^{(m)\cancel\frR}_i \ket{\sigma'}}{E(\sigma) - E(\sigma')}\left( 1 -\delta_{\sigma \sigma'} \right) , \label{def_loc}
\end{equation}
where we use the same notation as in subsection \ref{sec:treatment of resonant regions}. 
As already indicated above, we see that $A\dg{m}_i$ has support in $S\dg{m}_i$ because  $W^{(m)\cancel\frR}_i$ has support in $S\dg{m}_i$ and  $H_0$ is a sum of on-site terms.

\section{Inductive bounds}\label{sec: inductive bounds}

\subsection{Bounds on local terms}
We use a technical lemma from \cite{schur1911}, see \cite{HORN1992} for a modern presentation:
\begin{lemma} \label{lem: hadamard}
    Let $B,C \in M_n(\mathbb{C})$ and let $B\circ C$ be their Hadamard product, i.e. $(B\circ C)_{i,j} = B_{i,j} C_{i,j} $.  Then
    \begin{equation}
        ||B \circ C || \leq ||B|| \max_{i}{ \left(\sum_{j} |C_{i,j}|^2 \right)}^{1/2}.
    \end{equation}
\end{lemma}

We can now derive bounds on $A\dg{m}_i$ and $D\dg{m}_i$ in terms of the operators $F\dg{m}_i$. 
\begin{lemma} \label{lem: bound on a}
\begin{equation}
    ||A\dg{m}_i|| \leq   \frac{4}{\delta}   ||F\dg{m}_i||, \qquad 
  ||D\dg{m}_i ||  \leq   ||F\dg{m}_i ||       .
\end{equation}
\end{lemma}
\begin{proof}
Indeed, with a choice $B = W_i^{(m)\cancel{\frR}}$ and $C_{\sigma \sigma'} = \frac{1-\delta_{\sigma \sigma'}}{E_\sigma-E_{\sigma'}}$  in Lemma \ref{lem: hadamard} we obtain
\begin{equation}
    ||A_i\dg{m}|| \leq \sqrt{\frac{2}{\delta^2}\sum_{k=1}^\infty \frac{1}{k^2}} \times  ||W_i^{(m)\cancel{\frR}}|| \leq \frac{\pi}{\delta\sqrt{3}}||W_i^{(m)\cancel{\frR}}|| \leq \frac{2\pi}{\delta\sqrt{3}} ||F_i\dg{m}||.
\end{equation}
For the diagonal terms $D\dg{m}_i $, we similarly choose $B=F_i\dg{m}$ and $C=\text{id}$ in Lemma \ref{lem: hadamard}.
\end{proof}
From the second bound of Lemma \ref{lem: bound on a}, we then get
\begin{align}
   ||W_i^{(m)\cancel{\frR}}|| &\leq 2||F_i\dg{m}||,\label{eq:W^mNR bound} \\
    ||W_i^{(m)\frR}|| &\leq 2||F_i\dg{m}||. 
\end{align}

\subsection{Combinatorics}\label{sec: combi}

We will now inductively prove bounds on the local terms of $A\dg{m}$ and $F\dg{m}$. Certainly, one could envisage proving sharper bounds, i.e.\ by mimicking the reasoning in \cite{abanin2017rigorous} we would get $m!$ instead of $m!^3$, but such bounds would not lead to a substantially stronger result, as we will see in subsection \ref{sec: prob estimates}. 
\begin{lemma} 
Let $C_1=\tfrac{2}{\log(5/4)}$ and $K=\frac{4C_1R^2}{\delta}$
For any $m\geq 1$, 
\begin{equation}
    ||F\dg{m}_i|| \leq m!^3 K^{m-1}, \label{eq: F^m bound}
\end{equation}
and 
\begin{equation}
    ||A\dg{m}_i|| \leq  \frac{m!^3}{C_1 R^2}K^{m}. \label{A_bound}
\end{equation}
\end{lemma}

\begin{proof}
For any $m\geq 1$, \eqref{A_bound} follows from \eqref{eq: F^m bound} by Lemma \ref{lem: bound on a}. For $m=1$, \eqref{eq: F^m bound} follows immediately by the definition $F\dg{1}_i=V_i$ and the assumption $||V_i||\leq 1$. 
It remains hence to prove \eqref{eq: F^m bound} for $m>1$ given that  \eqref{A_bound} and \eqref{eq: F^m bound} hold true with $m$ replaced by $m'$ for $m'\leq m-1$.
We start from \eqref{eq: f expansion} where we use  $\sum_{i_0}[A\dg{m_1}_{i_1},(H_0)_{i_0}]=-W^{(m_1)\cancel\frR}_{i_1}$ to eliminate the rightmost commutator in the first term and the sum over $i_0$.   We bound the remaining commutators of local terms by  $||[B_i,B_j]|| \leq 2||B_i|| ||B_j||$ and using $|| W^{(m_1)\cancel\frR}_{i_1} || \leq 2|| F^{(m_1)}_{i_1} || $ (from \eqref{eq:W^mNR bound}).  The sums over $i_1,\ldots,i_k$ and $i_0,i_1,\ldots,i_k$ are then controlled by using that 
the local terms of $F\dg{m'}, W\dg{m'}, A\dg{m'}$ are supported on intervals of length $m'R$, as explained in Section \ref{section: Inductive diagonalization}. 
In this way, we obtain the following bound, for any $i$:
\begin{align}
    ||F\dg{m}_i|| &\leq  K^{m-1} \sum_{k = 2}^m \frac{2^{k}}{k!}  \frac{1}{(C_1R^2)^{k-1} } \sum_{\substack{ 1 \leq m_1, \ldots, m_k \leq m-1\\
    m_1 + \ldots + m_k = m}} N(m_1,\dots,m_k) \prod_{p=1}^k m_p!^3  \label{eq: inductionone}\\
    &+  K^{m-1} \sum_{k = 1}^{m - 1} \frac{2^k}{k!} \frac{1}{(C_1R^2)^k} \sum_{\substack{ 1 \leq m_1, \ldots, m_k \leq m - 1\\
    m_1 + \ldots + m_k = m-1}} N(1,m_1,\dots,m_k) \prod_{p=1}^k m_p!^3 \label{eq: inductiontwo},
\end{align}
where $N(m_1,\dots,m_s)$ is the  number of sequences of intervals $(I_1,\ldots, I_s)$ satisfying
\begin{enumerate}
    \item $|I_p| = m_p R$ for all $p=1,\ldots,s$.
    \item $I_j \cap \bigcup_{p=1}^{j-1} I_p \neq \emptyset$ for all $j=2,\ldots,s$.
    \item $\min(\bigcup_{p=1}^{s} I_p)=i$.
\end{enumerate}
These conditions originate from the structure of nested commutators and our choice to anchor the local term $F\dg{m}_i$  to start at site $i$. 
Let us abbreviate

\begin{equation}
  Z= \frac{1}{m!^3}\sum_{k = 2}^{m } \frac{2^k}{(k-1)!} \frac{1}{(C_1R^2)^{k-1}} \sum_{\substack{ 1 \leq m_1, \ldots, m_k \leq m\\
    m_1 + \ldots + m_k = m}} N(m_1,\dots,m_k) \prod_{p=1}^k m_p!^3 . \label{eq:lemma}
\end{equation}
The desired bound will be proven once we verify that
$$
Z\leq 1/2.
$$
Indeed, both the sum over $k$ in \eqref{eq: inductionone} and the sum over $k$ in \eqref{eq: inductiontwo}  are bounded by  $m!^3K^{m-1}Z$. We now state three estimates that will help to bound $Z$.  
\begin{enumerate}
    \item  \begin{align}
N(m_1,\ldots,m_k) &\leq 
 R^{k-1} (m_{1}+m_2)(m_{1}+m_{2}+m_3)\dots(m_1+\dots+m_k) \times R (m_1+\dots+m_k) \nonumber \\ 
&\leq mR^{k} \frac{m!}{(m-k+1)!}. \label{eq: bound on big n}
\end{align}
In the first line, the factors before $\times$ are due to the choice the intervals $I_p, p=2,\ldots,m$ assuming that $I_1$ has already been chosen, and the factor after $\times$ is due to the choice of first interval.
\item  Since $m_p\geq 1$, we have 
\begin{equation}
   \prod_{p=1}^k m_p! \leq \left(\sum_{p=1}^k m_p -k+1\right)! \leq   \left(m -k+1\right)! .
\end{equation}
\item The number of terms in the sum $\sum_{\substack{ 1 \leq m_1, \ldots, m_k \leq m\\
    m_1 + \ldots + m_k = m}} \ldots $ is bounded by the binomial coefficient 
    \begin{equation}
    \binom{m-1}{k-1}= \frac{(m-1)!}{(m-k)!(k-1)!} .
    \end{equation}
\end{enumerate}
Using these three estimates, we get the bound
\begin{align*}
 Z
    &\leq \frac{2R}{m!^3} \sum_{k = 2}^{m } \frac{(2/(C_1R))^{k-1}}{(k-1)!} m\frac{m!}{(m-k+1)!}\frac{(m-1)!}{(m-k)!(k-1)!}(m-k+1)!^3\\
    &
    \leq 2R \sum_{k = 2}^{m } \frac{(2/(C_1R))^{k-1}}{(k-1)!^2}\leq 2  \numberthis
    (e^{2/C_1}-1)
\end{align*}
where we used that $R\geq 1$.  For large enough $C_1$, this is smaller than $1/2$, which ends the proof.     
\end{proof}

\subsection{Bound on remainder $\delta H\dg{n}$}

The remainder is defined as 
$$
\delta H\dg{n}=\sum_{m=n+1}^\infty  J^m G\dg{m}.
$$

The local terms of $G\dg{m}$,  $G\dg{m}_i$, for $m>n$, are defined in the same way as the local terms of $F\dg{m}$ in  \eqref{sec: combi}.  We observe that $G\dg{m}_i$ have support in the intervals $[i,\ldots, i+Rm-1]$.  
By the same arguments that led to the bound in \eqref{eq:lemma}, we bound them as 
\begin{equation}
||G\dg{m}_i || \leq 2 K^{m-1}\sum_{k = 2}^{m } \frac{2^k}{(k-1)!} \frac{1}{(C_1R^2)^{k-1}} \sum_{\substack{ 1 \leq m_1, \ldots, m_k \leq n\\
    m_1 + \ldots + m_k = m}} N(m_1,\dots,m_k) \prod_{p=1}^k m_p!^3  .\label{eq:lemma again}
\end{equation}
We can now afford much cruder bounds, namely 
\begin{enumerate}
    \item  $N(m_1,\ldots, m_k) \leq (Rm)^k$, cf.\ the bound in \eqref{eq: bound on big n} .
    \item $\prod_{p=1}^k m_p!^3  \leq (n!)^{3k} $.
    \item The number of terms in the sum $\sum_{\substack{ 1 \leq m_1, \ldots, m_k \leq n\\
    m_1 + \ldots + m_k = m}}$ is bounded by $2^m$.
\end{enumerate}
We get then 
\begin{align*}
||G\dg{m}_i || \leq & 2 K^{m-1}\sum_{k =2}^{m} \frac{2^k}{(k-1)!} \frac{1}{(C_1R^2)^{k-1}} 2^m (Rm)^k n!^{3k}  \\
  \leq  &
\left(4 K^{m-1} n!^{3m} 2^m Rm  \right)  e^{\frac{2Rm}{C_1R^2} }  \\[2mm]
\leq  & 
\delta R Y^m, \qquad  Y= 4 e^{\frac{2}{C_1R}} K n!^{3}, \numberthis
\end{align*}
where in the last line, we used that $C_1R^2 >1$ and $m\leq 2^m$.

\section{Conclusion of the proof}\label{sec: Conclusion}
The proof of Theorems \ref{thm: quasi} and \ref{thm: current} follows from the inductive diagonalization scheme described above.

We introduce a parameter $\epsilon \in (0,1)$. From here onwards, we will use generic constants $C$ that can depend on the parameters $R,\beta,\epsilon$, but not the size of the chain $N$ or the coupling strength $J$, provided that $J$ is small enough.  
These generic constants $C$ can also change from line to line. 
We recall the definition
\begin{equation}
n_* =\left \lfloor\log( \frac{1}{J})^{1-\epsilon}\right\rfloor,    \label{n_choice again}
\end{equation}
and we will choose the expansion parameter $n$ that features in the previous sections, as $n=n_*$.

\subsection{Smallness of remainder terms and locality of the transformation $e^{\adjoint_A}$}

Then, we deduce the following result 
\begin{lemma}\label{lem: wrapup}
 There is a $J_*$ depending on the parameters $R$ and  $\epsilon,\beta$ such that, for $J\leq J_*$, we have 
 \begin{enumerate}
     \item For any $m$ satisfying $1\leq m\leq n_*$, 
$$  
  J^m||A_i\dg{m}|| \leq  e^{-\frac{1-\beta}{2}\log(1/J)m},
$$
and the same estimate holds with $A_i\dg{m}$ replaced by $D_i\dg{m}$.
\item The remainder term $\delta H\dg{n_*}$ is a sum of quasi-local terms
$$
\delta H\dg{n_*} =\sum_i  \delta H\dg{n_*}_i,\qquad
 \delta H\dg{n_*}_i= \sum_{m =n_*+1}^\infty  J^m G\dg{m}_i,
$$
such that $\delta H\dg{n_*}_i$ is 
$$ \left( e^{-\frac{1-\beta}{2} (\log(1/J))^{2-\epsilon} },\frac{1-\beta}{2}\log(1/J)\right) \quad \text{exp.\ quasi-local around $[i,i+(n_*+1)R-1]$}.$$ 
 \end{enumerate}
\end{lemma}

\begin{proof}
To prove item 1. we start from the bound  on  $||A\dg{m}_i||$ in \eqref{A_bound}, which reads, using a generic constant $C$,   $ J^m ||A\dg{m}_i|| \leq   J^{(1-\beta)m} C^m m^{3m} $.
To get the claimed estimate, it suffices hence to show that 
$$
C J^{(1-\beta)} m^3 \leq  e^{-\frac{1-\beta}{2}\log(1/J)}, \qquad  m \leq n_* .
$$
This follows indeed by choosing $J_*$ small enough and plugging the value \eqref{n_choice} for $n_*$.  The proof for $D\dg{m}_i$ is the same, but starting from \eqref{eq: F^m bound} instead of \eqref{A_bound}.
Item 2. is proven in an analogous way.
\end{proof}

To ensure that the transformation $e^{\adjoint_A}$ keeps local observables quasi-local, with good decay bounds, we rely on well-known Lieb-Robinson bounds. 
They allow us to establish the following estimate:
\begin{lemma}\label{lem: lr for a}
 Let $O_X$ be supported in an interval $X$. Then  $e^{\adjoint_A}(O_X)$ and $e^{-\adjoint_A}(O_X)$ are  $(|X| ||O_X||, \kappa ) $  exponentially quasi-local around $X$, with 
 $\kappa = \frac{1-\beta}{4R} \log(\frac{1}{J})$.
\end{lemma}
To prove this lemma, we use the conventions and terminology of \cite{Nachtergaele2019}. In particular,
we choose the so-called $F$-function
$$F(r)=\frac{e^{-2\kappa r}}{(1+r)^2}. $$ 
Then, by item 1. of Lemma \ref{lem: wrapup} above, the interaction $\Phi$ corresponding to our Hamiltonian  $\iu A$ has a bounded $||\Phi ||_{F} $-norm. 
The Lieb-Robinson bound stated in Theorem 3.1 in \cite{Nachtergaele2019} reads, for observables $O_X,O_Y$ supported in $X,Y$ 
$$
||[e^{\adjoint_A}(O_X),O_Y]|| \leq C ||O_X || \,||O_Y||    \sum_{x\in X, y\in Y} F(\distance(x,y)).
$$
To continue, we use a standard trick first introduced in \cite{bravyi2006lieb}; for any observable $B$, 
$$
||\tr_{Y}[B]- B ||  \leq \sup_{O_Y, ||O_Y||=1}    ||[B,O_Y]||    ,
$$
where $\tr_Y$ is the normalized partial trace over the region $Y$. This allows to express $B$ as a sum over local terms via
$$
B= \tr_{X^c}[B] + \sum_{r\geq 1} \left( \tr_{X_r^c}[B]-\tr_{X_{r-1}^c}[B]\right), 
$$
where $X_r=\{i \, | \, \mathrm{dist}(i,X)\leq r\}$.
Using this representation, the property of the partial trace $||\tr_{X^c} [B]|| \leq ||B||$, and the Lieb-Robinson bound, for the case where $X$ is an interval and $B=e^{\adjoint_A}(O_X)$, we conclude that $B$ is  
$(|X| ||O_X||,\kappa)$- exponentially quasi-local around $X$. Note that  we eliminated constants by weakening the decay to $\kappa$ instead of $2\kappa$ (the decay in the $F$-function) and choosing $J$ small enough. 
The statement for $e^{-\adjoint_A}$ follows by the same argument, since we only used bounds on the local terms in $A$.

\subsection{Construction of quasi-LIOMs}
 To construct quasi-LIOMs, we start analogously to the construction in subsection \ref{sec: lioms one}.
We define a first resonant set $\widetilde\frR$ as comprising all the resonant intervals
$$
\widetilde\frR= \mathop{\bigcup}\limits_{m=1}^{n_*} \mathop{\bigcup}\limits_{i: S\dg{m}_i  \, \text{resonant}}    S\dg{m}_i .
$$
 If we were to mimic the construction of subsection \ref{sec: lioms one}, we would now split $\widetilde\frR$ into connected components.  However, in contrast to the situation of subsection \ref{sec: lioms one}, the Hamiltonian $H'$ exhibited in \eqref{eq: result of scheme} contains the diagonal terms $D\dg{n_*}_i$ supported on intervals of length $n_*R$, which means that two resonant terms $W^{(m)\frR}_i, W^{(m')\frR}_{i'}$ have to be considered together whenever there is a diagonal term $D\dg{n_*}_{i''}$ that fails to commute with both of them, i.e.\ whenever the distance between the intervals $S\dg{m}_i$ and $S\dg{m'}_{i'}$ is smaller than $n_*R$.  
 
This inspires the following construction: 
We define connected components of the set $\widetilde\frR$ calling sites $i,j \in \widetilde\frR$ connected if the distance $|i-j|$  between them is smaller than $n_*R$. For each connected component we then find the smallest covering interval $T_\alpha$, see Figure \ref{fig:connected components}. This way we obtain a collection $(T_\alpha)$ such that $\widetilde\frR \subset \cup_\alpha T_\alpha$ and $\distance(T_\alpha,T_{\alpha'}) \geq n_* R$ for $\alpha\neq\alpha'$.

\begin{figure}[H]
    \centering
    \includegraphics[width=\textwidth]{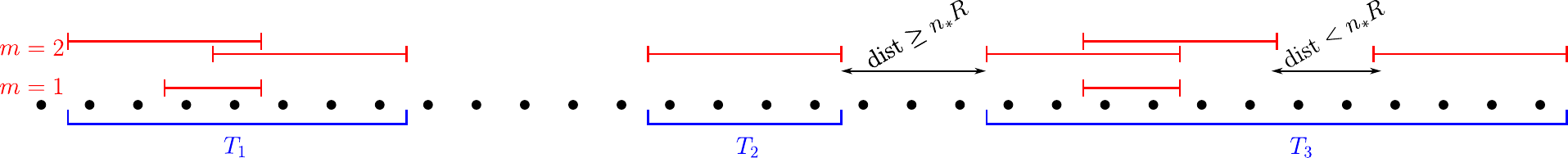}
    \caption{An example of resonant intervals $S\dg{m}_i$ for $m=1,2$. Here, we have chosen $R=2$ and $n_*=2$. The resulting collection of intervals $T_\alpha$ is shown.}
    \label{fig:connected components}
\end{figure}
The final resonant set is then defined as the union of $T_\alpha$;
$$
\frR=\bigcup_{\alpha}  T_\alpha .
$$
For any  $T_\alpha$, we define\begin{equation}
    \caE_\alpha = \sum_{i\in T_\alpha} \left( h_i Z_i + \sum_{m=1}^{n_*} J^m W_i^{(m)\frR}\right) + \sum_{m=1}^{n_*} \sum_{i: S_i\dg{m} \cap T_\alpha \neq \emptyset}  J^m D_i^{(m)} .
\end{equation}
The term between brackets is supported on $T_\alpha$. Indeed,  terms $W_i^{(m)\frR}$ with $i$ closer than distance $Rm$ to the maximum of $T_\alpha$ are zero, else the interval $T_\alpha$ would have extended further right. 
However, the terms $D_i\dg{m}$ extend beyond $T_\alpha$. Therefore, we define fattened intervals 
$$\tilde T_\alpha= [\min(T_\alpha) -n_*R +1,\max(T_\alpha)+n_*R-1],$$
such that $\caE_\alpha$ is indeed supported on $\widetilde T_\alpha$, see Figure \ref{fig: fattened}.
\begin{figure}[H]
    \centering
    \includegraphics[width=\textwidth]{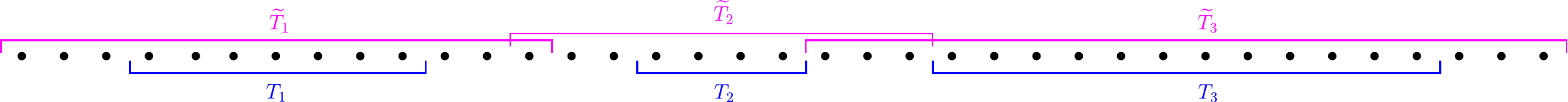}
    \caption{ We illustrate the fattened intervals $\widetilde T_\alpha$ starting from the example of Figure \ref{fig:connected components}.}
    \label{fig: fattened}
\end{figure}
We note that  $H'$ can now be  written as
\begin{equation}\label{eq: expression hprime}
H'= \sum_{i: i \not\in \frR} h_i Z_i + \sum_{m=1}^{n_*} \sum_{i:S_i\dg{m} \cap \frR=\emptyset}  J^m D\dg{m}_{i} + \sum_{\alpha}   \caE_\alpha   + \delta H\dg{n_*}  .
\end{equation}
The nice property of the operators $ \caE_\alpha$ is that they mutually commute for different $\alpha$, and they commute with all $Z_i, i \not\in \frR$.
In particular, for given $\alpha$, the operators $ \caE_\alpha$ and $Z_i, i \in \widetilde T_\alpha \setminus T_\alpha$ form a mutually commuting family, supported on $\widetilde T_\alpha$. These operators can hence be diagonalized together in $Z$-basis. We can then construct $|T_\alpha|$ spin operators (in the sense of Theorem \ref{thm: quasi}) $\tau'_{\alpha,j}, j=1,\ldots, |T_\alpha|$ such that $\tau'_j$ are supported on $\widetilde T_\alpha$ and such that they commute with each other and with $Z_i, i \in \widetilde T_\alpha \setminus T_\alpha$.  The labelling by $j$ is arbitrary. 
 We now set 
 \begin{equation}
      \tau'_{i} =  \begin{cases} \tau'_{\alpha,j}  &  i=\min (T_\alpha)+j-1  \\
        Z_i     &  i \not \in \frR
      \end{cases} .
 \end{equation}
 In this way we obtain a full set of $N$ spin operators $\tau'_i$ that mutually commute. 
 \begin{figure}[H]
    \centering
    \includegraphics[width=0.3\textwidth]{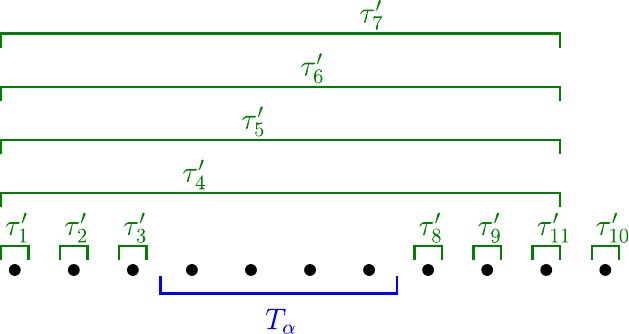}
    \caption{The support of quasi-LIOMs $\tau'_i$. We have chosen $R=2$ and $n_*=2$.}
    \label{fig:qLIOMs range}
\end{figure}
\subsection{Bounds involving quasi-LIOMs}

Finally, we are ready to define the $\tau_i$ spin-operators as
$$
\tau_i = e^{-\adjoint_A}(\tau_i')
$$
and to check some of the assertions of Theorem \ref{thm: quasi}. 
Items 1,2,5 of Theorem \ref{thm: quasi} follow immediately from the corresponding properties of the $\tau_i'$. 
For $i \not\in \frR$, we set $M_i=\{i\}$ and for $i\in \frR$, we set of course $M_i=\widetilde T_\alpha$ for the corresponding $\alpha$.   Then, Item 3 follows from the locality-preserving property of $e^{-\adjoint_A}$, i.e.\ Lemma \ref{lem: lr for a} and the fact that $||\tau_i'||=1$. 
To get item 4, we bound, for $i\not\in\frR$;
\begin{equation}\label{eq: comm with tau}
||[\delta H\dg{n_*},\tau'_i] || \leq  C (n_*+1)R e^{-\frac{1-\beta}{2} (\log(1/J))^{2-\epsilon} } .
\end{equation}
Indeed, if we consider only the leading contributions to $\delta H\dg{n_*}_i$, with support in intervals of length $(n_*+1)R$, then there are  $(n_*+1)R $ sites $i$ such that $\delta H\dg{n_*}_i$ contributes to the commutator, hence the factor $(n_*+1)R $ in the bound. One easily checks that the subleading terms in  $\delta H\dg{n_*}_i$ give a smaller contribution to the commutator, provided $J$ is sufficiently small, and they can be accounted for by the constant $C$ in front. 
For $i \in \frR$, we use analogous reasoning to derive the bound 
$$
||[\delta H\dg{n_*},\tau'_i] || \leq C (|\widetilde T_\alpha|+2(n_*+1)R) e^{-\frac{1-\beta}{2} (\log(1/J))^{2-\epsilon} }  .
$$
Taking $J$ small enough and recalling that  $|M_i|=|\widetilde T_\alpha|\geq 2n_*R$, we obtain now item 4. of the theorem. Note that we changed $\frac{1-\beta}{2}$ to $\frac{1-\beta}{4}$ to absorb polynomial factors in $\log(1/J)$. 

 \begin{figure}[H]
    \centering
    \includegraphics[width=\textwidth]{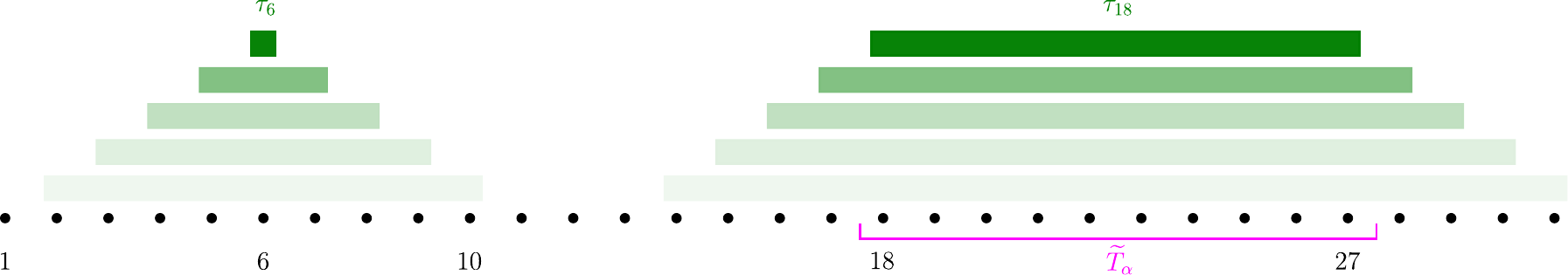}
    \caption{The exponential quasilocalty of quasi-LIOMs $\tau_i$. The spatial extend of rectangles illustrates the range of local terms, while their opaqueness illustrates their strength. The figure depicts the first $5$ strongest terms from the expansion of $\tau_6$ and $\tau_{18}$. Note, that $\tau_{18}$ corresponds to $\widetilde{T}_\alpha$ with length $|\widetilde{T}_\alpha|=10$, hence all quasi-LIOMs $\tau_i$ for $i\in[18,27]$ look the same in the above representation.}
    \label{fig:qLIOMs locality}
\end{figure}

\subsection{Probabilistic estimates}\label{sec: prob estimates}

In Section \ref{sec: cartoon}, we already computed the probability of a given bond $(i,i+1)$ being resonant at the first scale for $R=2$. For higher scales and general $R$, we stated the non-resonance condition \eqref{eq:non-resonance}.  The probability of \eqref{eq:non-resonance} being violated is estimated as
\begin{equation}
    \mathrm{Prob} (\, \min_{\eta\neq0} \, |\sum_{i\in S_i\dg{m}  } \eta_i h_i| < \delta ) \leq \sum_{\eta\neq 0} \mathrm{Prob} ( \, | \sum_{i\in S_i\dg{m}} \eta_i h_i \, | < \delta ) 
    \leq 3^{mR}\delta,
\end{equation}
where the prefactor $3^{mR}$ accounts for the number of choices of $\eta$. For a single $\eta$, we choose an index $j$ such that $\eta_j\neq 0$ and we write 
$|\sum_i\eta_i h_i|=|h_j +a|$ with $a$ independent of $h_j$. The probability that this expression is smaller than $\delta$ is then obviously bounded by $\delta$, giving the above estimate.

Therefore, the probability to have a resonant interval of length $mR$ starting at site $i$ is bounded by 
$$
3^{Rm}\delta ,
$$
and the probability of having a resonant interval - of arbitrary length - starting at $i$ is obtained by summing this from $m=1$ to $m=n_*$.
If a site $i$ is in the resonant set $\frR$, a resonant interval must have started somewhere in the interval $[i-2n_*R+1,i]$. This leads to the estimate
\begin{equation}\label{eq: prob estimate} 
\mathrm{Prob}(i \in \frR) \leq 2n_*R \sum_{m=1}^{n_*}  3^{Rm}\delta \leq   4n_*R 3^{Rn_*}\delta \leq J^{\beta/2} ,
\end{equation}
where the last bound follows by the  choice \eqref{n_choice} for $n_*$, for $J$ small enough.   
The starting of a resonant interval at a site in $[i-2n_*R+1,i]$ is a necessary condition, but not a sufficient condition, for the event $i\in \frR$ to be true.  To determine whether $i\in \frR$, we need information about resonant intervals starting in $[i-2n_*R+1,i+n_*R-1]$. This means we need information on the random variables $h_j$ for $j \in  [i-2n_*R+1,i+2n_*R-1]$. 
We have now hence proven item 6. of Theorem \ref{thm: quasi}.  The first claim of item 7. is a direct consequence of the way the fattened intervals $\widetilde T_\alpha$ were defined. The bound in item 7. can be deduced in an elementary way: 
In order for $|M_i|=|\widetilde T_\alpha|\geq \ell$, we need that at least  $\ell/(4n_*R)$ sites, spaced by $4n_*R$, are in $\frR$. Since these are independent events, the probability of an interval $|\widetilde T_\alpha|\geq \ell$ starting at a given site is bounded by (using item 6)
$$
J^{(\beta/2)(\ell/(4n_*R))} \leq   e^{-\frac{\beta}{8R} \log (1/J)^{\epsilon}\ell} .$$
However, for $|M_i|\geq \ell$, this interval can start at any site in $[i-\ell+1,i+\ell-1]$, so we multiply the above probability by $2\ell$. By readjusting constants, we then get item 7. 

It is interesting to note that the requirement that $\mathrm{Prob}(i \in \frR) \leq J^{\beta/2}$  forces us to let $n_*$ grow not faster then logarithmically in $J$. This is the reason that, as things stand, it does not make much sense to improve the combinatorial estimates in Section \ref{sec: inductive bounds}.

\subsection{Proof of Theorem \ref{thm: current}}
For any site $i$, we define the random variable 
$$
Q(i)=\sum_\alpha |\widetilde T_\alpha|e^{-\frac{1-\beta}{2}\log(1/J) \mathrm{dist}(i,\widetilde T_\alpha)}
$$
that will play a major role in what follows. Let $x$ be the site at smallest distance to the site $i_*$, such that 
\begin{equation} \label{eq: requirement}
Q(x) \leq  e^{-\frac{1-\beta}{2}\log(1/J)}  .
\end{equation}
Then
\begin{lemma}\label{lem: random q}
For $J$ sufficiently small, 
$$\mathrm{Prob}(\mathrm{dist}(x,i_*) \geq \ell)  \leq e^{-c_0(\log(1/J))^{\epsilon}\ell}, \qquad c_0= 10^{-3} \beta/R  .$$
\end{lemma}
Note in particular, that there is no $\alpha$ such that $ x\in \widetilde T_\alpha$. Else, $\mathrm{dist}(x,\widetilde T_\alpha)=0$ and the requirement \eqref{eq: requirement} would not be satisfied. The proof of Lemma \ref{lem: random q} rests entirely on the sparsity of the resonant set $\frR$. Since it does not add any further insight, we postpone it to the Appendix.

Let us now come to Theorem \ref{thm: current}. In the setup, we defined a splitting
$H=H_L+H_R$ based on a fiducial point $i_*$.  It is advantageous now to vary the fiducial point. We write therefore $H=H_{L,i}+H_{R,i}$ for any point $i$, with  $H_{L,i}= \sum_{j \leq i} (h_j Z_j + JV_j)$, cf.\ the expression in \eqref{eq: naive splitting}. In particular, we will choose the random site $x$, defined above, as the fiducial point. 

We start from the expression for $H'$ given in \eqref{eq: expression hprime} and we define a left-right splitting (relative to the site $x$) of this Hamiltonian as 
\begin{equation}\label{eq: definition prime splitting}
  H_{L,x}'=
  \sum_{\substack{i: i \leq x\\
    i \not\in\frR}}
h_i Z_i +  \sum_{m=1}^{n_*}  
  \sum_{\substack{i: i \leq x\\
     S_i\dg{m} \cap \frR =\emptyset}}
J^m D\dg{m}_i + \sum_{\alpha: \min \widetilde T_\alpha \leq x } \caE_\alpha  + \sum_{i: i \leq x} \delta H\dg{n_*}_i  
\end{equation}
and $H_{R,x}'=H'-H_{L,x}'$. 
Now we split
\begin{equation}\label{eq: splitting acro}
   e^{\adjoint_A}[H_{L,x},H]
= [e^{\adjoint_A}(H_{L,x})-H_{L,x}',H'] + [H_{L,x}',H']  .
\end{equation}
Let us start by estimating the commutator $[H_{L,x}',H']$  on the right-hand side.
\begin{lemma}\label{lem: small comm}
For sufficiently small $J$, 
   $$  || [H_{L,x}',H'] || \leq    e^{-\frac{1-\beta}{4} (\log(1/J))^{2-\epsilon} }.$$ 
\end{lemma}
\begin{proof}
We write $K_j$ to denote
 local terms of the type $h_jZ_j, J^m D_i\dg{m}$ 
 in the expression \eqref{eq: expression hprime}, so that this expression consists of the contributions
\begin{enumerate}
    \item  $ \mathop{\sum}\limits_{j\leq x, i>x}   [K_j,\delta H\dg{n_*}_i]$. 
    \item  $\mathop{\sum}\limits_{j\leq x, i>x}  [\delta H\dg{n_*}_j, K_i]$.
    \item $\mathop{\sum}\limits_{j\leq x, i>x}  [\delta H\dg{n_*}_j, \delta H\dg{n_*}_i]$. 
\end{enumerate}
Then there are also the contributions involving $\caE_\alpha$, they are of the form
\begin{enumerate}
    \item[4.]  $  
    \mathop{\sum}\limits_{  
    \substack{ i>x \\ \alpha:  \min \widetilde T_{\alpha}\leq x
     }}
    [\caE_\alpha,\delta H\dg{n_*}_i].$ 
    \item[5.]  $
        \mathop{\sum}\limits_{  
    \substack{ i\leq x \\ \alpha: \min\widetilde T_{\alpha}>x
     }}  
    [\delta H\dg{n_*}_i,\caE_\alpha]$.
\end{enumerate}
The contributions 1.,2.,3. together are bounded by 
$$
 C((n_*+1)R)^2 e^{-\frac{1-\beta}{2} (\log(1/J))^{2-\epsilon} } ,
$$
where we used the same reasoning as for the bound \eqref{eq: comm with tau}.   Similarly, contributions 4.,5. are bounded together by 
$$
C Q(x) ((n_*+1)R) e^{-\frac{1-\beta}{2} (\log(1/J))^{2-\epsilon} }, 
$$
where $Q(x)$ is the random variable defined above Lemma \ref{lem: random q}.   To get this estimate, we use that $||\caE_\alpha|| \leq C |\widetilde T_\alpha |$ and that there is no interval $\widetilde T_\alpha$ containing $x$, see the remark following Lemma \ref{lem: random q}.  
\end{proof}

We now move to the first term on the right-hand side of \eqref{eq: splitting acro}.

\begin{lemma}\label{lem: obs o prime}
There is an observable $O'$ with 
$$
|| O'|| \leq   C J^{\frac{1-\beta}{4R}} 
$$
satisfying 
  $$  [e^{\adjoint_A}(H_{L,x})-H_{L,x}',H'] = [O',H'] .
$$
\end{lemma}
\begin{proof}
Let $\tr_{>x}$ be the normalized partial trace over the region $\{i \,|\, i > x\}$. Then we have, since $e^{\adjoint_A}(H)=H'$,
\begin{equation}\label{eq: full partial is zero}
    0=\tr_{>x} [e^{\adjoint_A}(H)] - \tr_{>x} [H'] .
\end{equation}
By a standard application of Lieb-Robinson bounds (see the discussion following Lemma \ref{lem: lr for a}), we have
\begin{equation}\label{eq: proof of split one}
||e^{\adjoint_A}(H_{L,x})-\tr_{>x} [e^{\adjoint_A}(H_{L,x})] || \leq  C J^{\frac{1-\beta}{4R}} ,
\end{equation}
for $J$ small enough. 
By similar reasoning, we get 
\begin{equation}\label{eq: proof of split two}
||\tr_{} [e^{\adjoint_A}(H_{R,x})]-\tr_{>x} [e^{\adjoint_A}(H_{R,x})] || \leq C J^{\frac{1-\beta}{4R}} ,
\end{equation}
where $\tr[\cdot]$ is the normalized trace over the whole chain. 
We can also estimate
\begin{equation}\label{eq: proof of split three}
\tr_{>x} [H'] - H_{L,x}'  = \tr [H'_{R,x}] + \sum_l (\tr_{>x} [A_l]-A_l)  ,
\end{equation} 
where the sum over $A_l$ runs over all local terms included in the definition of $H'_L$, in \eqref{eq: definition prime splitting},  but with support extending to the right of $x$. By the remark following Lemma \ref{lem: random q}, none of these local terms is of the form $\caE_\alpha$, and so they are all of the form $J^mD\dg{m}_i$. Therefore, by Lemma \ref{lem: wrapup}, we have good bounds on all the $A_l$ and we can estimate
\begin{equation}\label{eq: proof of split four}
 \sum_l  || \tr_{>x} [A_l]-A_l ||  \leq 
  2\sum_l  || A_l || \leq C J^{\frac{1-\beta}{2}} .
\end{equation} 

Starting from \eqref{eq: full partial is zero}
and using (\ref{eq: proof of split one}, \ref{eq: proof of split two}, \ref{eq: proof of split three} and \ref{eq: proof of split four})
we get 
$$
|| e^{\adjoint_A}(H_{L,x})-  H_{L,x}' + a || \leq C J^{\frac{1-\beta}{4R}},\qquad   a=\tr_{} [e^{\adjoint_A}(H_{R,x})]- \tr [H_{R,x}'] .
$$
Setting $O'= e^{\adjoint_A}(H_{L,x})-  H_{L,x}' + a  $ and using that $a$ is a number, i.e.\ it commutes with all operators, we get the claim of the lemma.
\end{proof}

We turn now to the proof of Theorem \ref{thm: current}.   Let us apply $e^{-\adjoint_A}$ to \eqref{eq: splitting acro} and use Lemma \ref{lem: obs o prime}  and Lemma \ref{lem: small comm}  obtaining
\begin{equation}\label{eq: splitting acro conclusion}
  || [H_{L,x},H] - [e^{-\adjoint_A} (O'),H] ||  \leq   e^{-\frac{1-\beta}{4} (\log(1/J))^{2-\epsilon} }
\end{equation}
To get  Theorem \ref{thm: current}, we have to go back to the fiducial site $i_*$, To that order, we remark that 
$$
H_{L,i_*}-H_{L,x}=O'', \qquad ||O''|| \leq \mathrm{dist}(x,i_*) (1+J).
$$
Setting now 
$$
O=O''+ e^{-\adjoint_A} (O')
$$
we obtain the claim of Theorem \ref{thm: current}.

\section*{Acknowledgements} \label{sec: acknowledgements}

 W.D.R. and O.P. were supported in part by the FWO (Flemish Research Fund) under grant
G098919N.

\section*{Appendix} \label{sec: appendix}

In this appendix, we prove of Lemma \ref{lem: random q}.
Let us abbreviate 
$$
\kappa= \frac{1-\beta}{2}\log (\frac{1}{J}), \qquad \nu=  \frac{\beta}{8R} \left(\log(\frac{1}{J})\right)^\epsilon ,
$$
so that we can rewrite
$$
Q(i)=\sum_\alpha |\widetilde{T}_\alpha|e^{-\kappa \mathrm{dist}(i, \widetilde{T}_\alpha)} .
$$

We need the following auxiliary result 
\begin{lemma}\label{lem: defining intervals}
If $Q(i) > M$ for some site $i$ and $M>0$, then there exists $\alpha$ such that 
\begin{equation}\label{eq: bassin equation}
|\widetilde T_\alpha|e^{-\kappa \mathrm{dist}(i,\widetilde T_\alpha)/2}  \geq \frac{M}{2} (1-e^{-\kappa/2}) .
\end{equation}
\end{lemma}

\begin{proof}
If there were no such $\alpha$, we dominate
 $$
 Q(i) < \frac{M}{2} (1-e^{-\kappa/2})  \sum_\alpha  e^{-\kappa \mathrm{dist}(i,\widetilde T_\alpha)/2} \leq M (1-e^{-\kappa/2}) \sum_{j=0}^\infty  e^{-\kappa j/2}  \leq (1-e^{-\kappa/2})  \frac{M}{1-e^{-\kappa/2}} .
 $$
\end{proof}

We use Lemma \ref{lem: defining intervals} with $M=e^{-\kappa}$. If site $i$ satisfies the inequality \eqref{eq: bassin equation} for some $\alpha$, then, for $\kappa$ large enough (hence $J$ small enough),
\begin{align} 
   \mathrm{dist}(i,\widetilde T_\alpha) & \leq \frac{2}{\kappa} \log |\widetilde T_\alpha| - \frac{2}{\kappa} \log (\frac{e^{-\kappa}}{2} (1-e^{-\kappa/2}))  \nonumber\\
    & \leq \frac{2}{\kappa} \log |\widetilde T_\alpha| +3   \leq 4 |\widetilde T_\alpha|. \label{eq: bassin two}
\end{align}

We recall that Lemma \ref{lem: random q} depends on a fiducial site $i_*$ and we define intervals centered on $i_*$:  $$
B_L = \{ j:  \mathrm{dist}(i_*,j) \leq L/2\}.
$$
We will prove below that
\begin{equation}\label{eq: repetition lemma}
 \mathrm{Prob}(\forall i \in  B_{2\ell} : Q(i)\geq e^{-\kappa} )  \leq   e^{-\frac{1}{72}\nu \ell} , \qquad \text{for} \,  \ell>16 n_* R.
\end{equation}
 For $\ell \leq 16 n_* R$, the same bound follows by similar but simpler considerations as those presented below. In particular, in \textbf{Case 2} below, it suffices to exhibit a single site in $\frR \cap B_{2\ell}$ and use the bound from Theorem \ref{thm: quasi} item 6).  Once \eqref{eq: repetition lemma} is proven for all $\ell$, Lemma \ref{lem: bound on a} follows directly.

Let $E_L$ be the event that there is an interval $K$ of length $L$ such that $K \subset \frR$ and $\distance(K,i_*)\leq 14L$. Then 
$$
\mathrm{Prob}(E_L)\leq 29L e^{-\nu L}
$$
by the discussion in subsection \ref{sec: prob estimates}.
We now prove \eqref{eq: repetition lemma} by distinguishing two cases, assuming that the event in \eqref{eq: repetition lemma} holds.

\noindent \textbf{Case 1}:  The event $E_L$ holds for some $L>\ell/16$.
The probability of this is bounded by $29(\ell/16+1) e^{-\nu(\ell/16-1)}\leq e^{-\frac{1}{72}\nu \ell}$, and so in this case \eqref{eq: repetition lemma} is proven.

\noindent \textbf{Case 2}:  The event $E_L$ does not hold for any $L\geq \ell/16$.
Consider a site  $ j\in B_{\ell/8}$. Since it satisfies $Q(j) > e^{-\kappa}$, there must be at least one $\alpha$ such that $\mathrm{dist} (j,\widetilde{T}_\alpha) \leq 4|\widetilde{T}_\alpha|$, by \eqref{eq: bassin two}, and hence $\mathrm{dist} (i_*,T_\alpha)\leq 4|\widetilde{T}_\alpha| +n_*R+ \ell/16$. 
Since $E_L$ does not hold for any $L\geq \ell/16$ (by assumption), any such $\alpha$ satisfies $\widetilde{T}_\alpha \subset B_{2\ell}$ as $|\widetilde{T}_\alpha|<|T_\alpha|+2n_*R$ and $\ell > 16 n_* R$. For any such $\widetilde{T}_\alpha$ there exist at most $9|\widetilde{T}_\alpha|$ sites satisfying the equation \eqref{eq: bassin two}. On the other hand this equation has to be satisfied for all sites in $B_{\ell/8}$. Therefore, we conclude
$$
\sum_{\alpha: \widetilde{T}_\alpha \subset B_{2\ell}}  9 |\widetilde{T}_\alpha| \geq  |B_{\ell/8}|.
$$
This leads to the lower bound  $$ |\frR \cap B_{2\ell}| \geq \frac{\ell}{72}. $$
By the discussion in subsection \ref{sec: prob estimates}, this means that there are least $\frac{\ell/72}{4n_* R}$ sites in $\frR$, sufficiently spaced so as to be independent. 
We conclude that the probability of this occurring is bounded by $e^{-\frac{1}{72}\nu \ell }$.

\bibliographystyle{ieeetr}
\bibliography{LIOM}
\end{document}